\documentclass[manuscript,screen]{acmart}
\makeatletter
\let\@authorsaddresses\@empty
\makeatother
\usepackage{blindtext}

\usepackage{xpatch}

\makeatletter
\xpatchcmd{\ps@firstpagestyle}{Manuscript submitted to ACM}{}{\typeout{First patch succeeded}}{\typeout{first patch failed}}
\xpatchcmd{\ps@standardpagestyle}{Manuscript submitted to ACM}{}{\typeout{Second patch succeeded}}{\typeout{Second patch failed}}    \@ACM@manuscriptfalse
\makeatother
\usepackage{thm-restate}
\usepackage{graphicx}
\usepackage{hyperref}
\usepackage{tikz}
\usepackage{amsmath}
\usepackage{makecell}
\usepackage[normalem]{ulem}

\usepackage{blindtext}
\usepackage[normalem]{ulem}
\usepackage[hide]{chato-notes}
\usepackage{color}
\usepackage{graphics}
\usepackage{graphicx}
\usepackage{mathtools}
\usepackage{xspace}
\usepackage[noend]{algorithmic}
\algsetup{linenosize=\small}
\usepackage[ruled,vlined,linesnumbered,noend,nofillcomment]{algorithm2e}
\usepackage{epsfig}
\usepackage{epstopdf}
\usepackage{booktabs}
\usepackage[font={bf,small}]{caption}
\DeclareCaptionType{copyrightbox}
\usepackage{xcolor}
\usepackage{balance}
\usepackage{enumitem}
\usepackage{framed}
\usepackage{multirow}
\pdfmapfile{=libertine.map}
\usepackage{hyperref}

\hypersetup{
	colorlinks,breaklinks,
	urlcolor=[rgb]{0,0.25,0.75},
	linkcolor=[rgb]{0,0.25,0.75},
	citecolor=[rgb]{0,0.25,0.75}
}

\renewcommand\footnotetextcopyrightpermission[1]{} 
\DeclareMathOperator*{\argmax}{arg\,max}
\DeclareMathAlphabet\mathbfcal{OMS}{cmsy}{b}{n}

\newcommand{\fixedspaceword}[2][1]{%
  \begingroup
  \spaceskip=#1\fontdimen2\font
  \xspaceskip=0pt\relax 
  #2%
  \endgroup
}

\newcommand{\LL}[1]{{\color{black} #1}}

\newcommand{\laksRev}[1]{{\color{black} #1}}

\newcommand{\weic}[1]{{\color{blue} [\text{Wei Chen: }  #1]}}
\newcommand{\prib}[1]{{\color{red} [\text{Prithu: }  #1]}}

\newcommand{\pink}[1]{\textcolor{magenta}{#1}}

\newcommand{\spara}[1]{\vspace{1mm}\noindent\textbf{#1.}}

\newcommand{\eat}[1]{}

\newcommand{\InFullOnly}[1]{}
 
\theoremstyle{plain}

\newtheorem{theorem}{Theorem}
\newtheorem*{theorem-non}{Theorem}

\newtheorem{definition}{Definition}

\newtheorem{lemma}{Lemma}

\newtheorem{problem}{Problem}
\newtheorem{property}{Property}
\newtheorem{proposition}{Proposition}
\newtheorem{sub-proposition}{Sub-proposition}
\newtheorem{example}{Example}

\def\noflash#1{\setbox0=\hbox{#1}\hbox to 1\wd0{\hfill}}

\newcommand{\E}{{\mathbb{E}}\xspace}

\newcommand{\iaa}{{iPhoneX} \xspace}
\newcommand{\iab}{{iPad} \xspace}
\newcommand{\iac}{{AirPod} \xspace}

\newcommand{\ia}{{i_1}\xspace}
\newcommand{\ib}{{i_2}\xspace}

\newcommand{\util}{\mathcal{U}\xspace}

\newcommand{\ua}{{u}\xspace}

\newcommand{\price}{\mathcal{P}\xspace}
\newcommand{\val}{\mathcal{V}\xspace}
\newcommand{\noise}{\mathcal{N}\xspace}

\newcommand{\ps}{{ps}\xspace}
\newcommand{\cs}{{c}}
\newcommand{\ga}{{g_1}\xspace}
\newcommand{\gb}{{g_2}\xspace}
\newcommand{\gc}{{g_3}\xspace}
\newcommand{\ebay}{{eBay}\xspace}
\newcommand{\cgl}{{Craigslist}\xspace}
\newcommand{\fb}{{Facebook}\xspace}

\newcommand{\power}{\mathbb{P}}

\newcommand{\squishlist}{
 \begin{list}{$\bullet$}
  {  \setlength{\itemsep}{0pt}
     \setlength{\parsep}{3pt}
     \setlength{\topsep}{3pt}
     \setlength{\partopsep}{0pt}
     \setlength{\leftmargin}{2em}
     \setlength{\labelwidth}{1.5em}
     \setlength{\labelsep}{0.5em}
} }
\newcommand{\squishlisttight}{
 \begin{list}{$\bullet$}
  { \setlength{\itemsep}{0pt}
    \setlength{\parsep}{0pt}
    \setlength{\topsep}{0pt}
    \setlength{\partopsep}{0pt}
    \setlength{\leftmargin}{2em}
    \setlength{\labelwidth}{1.5em}
    \setlength{\labelsep}{0.5em}
} }

\newcommand{\squishdesc}{
 \begin{list}{}
  {  \setlength{\itemsep}{0pt}
     \setlength{\parsep}{2pt}
     \setlength{\topsep}{2pt}
     \setlength{\partopsep}{0pt}
     \setlength{\leftmargin}{2em}
     \setlength{\labelwidth}{1.5em}
     \setlength{\labelsep}{0.5em}
} }

\newcommand{\squishdesctight}{
 \begin{list}{}
  {  \setlength{\itemsep}{0pt}
     \setlength{\parsep}{0pt}
     \setlength{\topsep}{0pt}
     \setlength{\partopsep}{0pt}
     \setlength{\leftmargin}{1em}
     \setlength{\labelwidth}{1.5em}
     \setlength{\labelsep}{0.5em}
} }

\newcommand{\squishnumlist} {
\newcounter{qcounter}
\begin{list}{\arabic{qcounter}.~}{\usecounter{qcounter}} 
{  \setlength{\itemsep}{0pt}
    \setlength{\parsep}{0pt}
    \setlength{\topsep}{0pt}
    \setlength{\partopsep}{0pt}
    \setlength{\leftmargin}{2em}
    \setlength{\labelwidth}{1.5em}
    \setlength{\labelsep}{0.5em}
}}

\newcommand{\squishend}{
  \end{list}
}

\newcommand{\user}{{u}\xspace}

\newcommand{\itemset}{{I}\xspace}

\newcommand{\parameterset}{{\sf Param}\xspace}
\newcommand{\comic}{Com-IC\xspace}
\newcommand{\model}{UIC\xspace}

\newcommand{\allitems}{\mathbf{I}}
\newcommand{\allalloc}{\mathbfcal{S}}

\newcommand{\allseeds}{S}
\newcommand{\alliseeds}{S_i}

\newcommand{\utilw}{{\it util_{W}}\xspace}
\newcommand{\utilow}{\util_{W^N}}

\newcommand{\greedSeeds}{S^{\it Grd}}
\newcommand{\greedESeeds}{S^{\it GrdE}}

\newcommand{\greedAlloc}{\mathbfcal{S}^{\it Grd}}

\newcommand{\optAlloc}{\mathbfcal{S}^{\it OPT}}

\newcommand{\sw}{\rho}
\newcommand{\sww}{\rho_{W^N}}
\newcommand{\allblocks}{\mathcal{B}}

\newcommand{\OPT}{{\it OPT}}

\newcommand{\aware}{\mathcal{R}}
\newcommand{\awares}{\mathcal{R}^{\allalloc}}
\newcommand{\awarews}{\mathcal{R}_{W}^{\allalloc}}
\newcommand{\adopt}{\mathcal{A}}
\newcommand{\adopts}{\mathcal{A}^{\allalloc}}
\newcommand{\adoptws}{\mathcal{A}_{W}^{\allalloc}}

\newcommand{\ow}{W^N}
\newcommand{\dow}{\mathcal{D}^N}

\newcommand{\Iow}{\mathbf{I}_{\ow}^*}
\newcommand{\Istar}{\mathbf{I}^*}
\newcommand{\budgetSwitch}{{\it budgetSwitch}}
\newcommand{\true}{{\bf true}}
\newcommand{\false}{{\bf false}}

\newcommand{\bvec}{\vec{b}} 
\newcommand{\bmax}{\overline{b}}

\newcommand{\bbude}{e}

\newcommand{\qao}{q_{i_1|\emptyset}\xspace}
\newcommand{\qbo}{q_{i_2|\emptyset}\xspace}
\newcommand{\qab}{q_{i_1|i_2}\xspace}
\newcommand{\qba}{q_{i_2|i_1}\xspace}

\newcommand{\chgdel}[1]{\textcolor{red}{\sout{#1}}}
\newcommand{\chgins}[1]{\textcolor{black}{#1}}

\newcommand{\algo}{\textsf{bundleGRD}}
\newcommand{\numblocks}{t}

\newcommand{\PRIMM}{\textsf{PRIMA}}
\newcommand{\nodeselect}{\textsf{NodeSelection}}

\newcommand{\dbBook}{\mbox{Douban-Book}\xspace}
\newcommand{\dbMovie}{\mbox{Douban-Movie}\xspace}

\newcommand{\twit}{\mbox{Twitter}\xspace}
\newcommand{\orkut}{\mbox{Orkut}\xspace}
\newcommand{\db}{{Douban}\xspace}
\newcommand{\flix}{{Flixster}\xspace}

\newcommand{\RRSIM}{\textsf{RR-SIM}^{+}}
\newcommand{\RRCIM}{\textsf{RR-CIM}}
\newcommand{\id}{\textsf{item-disj}}
\newcommand{\bd}{\textsf{bundle-disj}}
\newcommand{\tcsc}{\textsf{BDHS-Concave}}
\newcommand{\tcss}{\textsf{BDHS-Step}}

\newcommand{\ind}{d_{in}}

\newcommand\ceil[1]{\lceil#1\rceil}
\begin{document}

\title{Maximizing Welfare in Social Networks under a Utility Driven Influence Diffusion Model}
\renewcommand{\shorttitle}{Utility Driven Influence Maximization}

\author{
Prithu Banerjee $\dag$ \hspace{4mm}  Wei Chen $\ddag$  \hspace{4mm}   Laks V.S. Lakshmanan $\dag$}

\affiliation{\\University of British Columbia, {\sf \{prithu,laks\}@cs.ubc.ca} $^\dag$, Microsoft Research, {\sf weic@microsoft.com} $\ddag$}

\renewcommand{\textrightarrow}{$\rightarrow$}

\begin{abstract}
Motivated by applications such as viral marketing, 
the problem of influence maximization (IM) has been extensively studied in the literature.
The goal is to select a small number of users to adopt an item such that it results in a large cascade of adoptions by others. Existing works have three key limitations. (1) They do not account for economic considerations of a user in buying/adopting items. (2) Most studies on multiple items focus on competition, with complementary items receiving limited attention. (3) For the network owner, maximizing social welfare is important to ensure customer loyalty\laksRev{, which is not addressed in prior work in the IM literature.}  In this paper, we address all three limitations and propose a novel model called UIC that 
combines utility-driven item adoption with influence propagation over networks.
Focusing on the mutually complementary setting, \laksRev{we formulate the problem of social welfare maximization in this novel setting.} We show that while the objective function is neither submodular nor supermodular, surprisingly a simple greedy allocation algorithm achieves a factor of $(1-1/e-\epsilon)$ of the optimum expected social welfare. We develop \textsf{bundleGRD}, a scalable version of this approximation algorithm, and demonstrate, with comprehensive experiments on real and synthetic datasets, that it significantly outperforms all baselines.

\end{abstract}
 



\maketitle

{\fontsize{9pt}{9pt} \selectfont
\textbf{Reference note:}\\
An abridged version of the paper appeared in 2019 International Conference on Management of Data (SIGMOD'19), June 30--July 5, 2019, Amsterdam, Netherlands. ACM, New York, NY, USA, 18 pages. https://doi.org/10.1145/3299869.3319879
} 

\onecolumn 


\section{Introduction}\label{sec:intro}
\eat{ 
flow and points to highlight in the intro. 
\begin{itemize} 
\item economic basis for item adoption; rooted in 
valuations and utility of itemsets. 

\item complementary (as opposed to competing) 
campaigns; very few such models before us. (be sure to 
check the literature for possible other works on 
complementary stuff.) 

\item simple model; easily extends to any number of items; 
simple algorithm; value specification relies on value 
oracles in most works on auction theory. unlike them, 
we do NOT need to know the value (or price or noise 
distributions) of items. remarkable. 

\item objective function mono but not sub/super mod; yet 
we can achieve a $(1-1/e-\epsilon)$-approx. -- first result 
of this kind to our knowledge. a short story of how we 
achieve this. 

\item algo. relies on item prefix property -- requires 
extension of existign scalable approx. algorithms like 
IMM. can we say anything abt extending SSA so as to satisfy 
the item prefix property? 

\item what do the experiments show? 

\end{itemize}
}

Motivated by applications such as viral marketing, 
the problem of influence maximization has been extensively studied in the literature \cite{li-etal-im-survey-tkde-2018}. The seminal paper of Kempe et al. \cite{kempe03} formulated \emph{influence maximization} (IM) as a discrete optimization problem: given a directed graph $G = (V,E,p)$, with nodes $V$, edges $E$, a function $p:E\rightarrow [0,1]$ associating influence weights with edges, a stochastic diffusion model $M$, and a seed budget $k$, select a set $S\subset V$ of up to $k$ seed nodes such that by activating the nodes $S$, the expected number of nodes of $G$ that get activated under $M$ is maximized. Two fundamental diffusion models are independent cascade (IC) and linear threshold (LT) \cite{kempe03}. Most of the work on IM has focused on a single item or phenomenon propagating through the network, and has developed efficient and scalable heuristic and approximation algorithms for IM \cite{ChenWW10, ChenWW10b, borgs14, tang15}. 


Subsequent work studies multiple campaigns propagating through a network \cite{HeSCJ12,BudakAA11,PathakBS10,BharathiKS07,lu2013,lu2015, chalermsook2015social}, mostly focusing on competing campaigns. 
\eat{Among these, the majority of studies have concentrated on items, products, or campaigns {\sl competing} with each other (see \textsection\ref{sec:related} for more details).}  
One exception 
is the \comic model by Lu et al. \cite{lu2015}, which studied the effect of {\sl complementary} products propagating through a network. A significant omission from the literature on IM and viral marketing is a study with item adoptions grounded in a sound economic footing.  
\eat{, where users adopt itemsets \sout{in a rational manner so as to} that maximize their own utility.} 

Adoption of items by users is a well-studied concept in economics \cite{myerson1981optimal, nisan2007}: item adoption by a user is driven by the \emph{utility} that the user can derive from the item (or itemset). Precisely, a user's utility for an item(set) is the difference between the \emph{valuation} that the user has for the item(set) and the \emph{price} she pays. A rich body of literature in combinatorial auctions (e.g., see~\cite{feige-vondrak-demand-2010, kapraov-etal-greedy-opt-soda-2013, korula-etal-online-swm-arxiv-2017}) studies the optimal allocation of goods to users, given the users' valuation for various sets of goods. 
\chgins{These studies are not concerned with the influence propagation in networks, whereby users' desire of items arises due to  
	the influence from their network neighbors who already adopted items, and then these users may in turn adopt
	the items if they could obtain positive utility from them and start influencing their neighbors about these items.
Considering such network propagation is important for applications such as viral marketing \cite{kempe03}.
}
%

\eat{ 
\weic{The above paragraph used to refer to ``network effect'', which is unclear and may have made the reviewers 
	thinking of network externality. I changed it to ``influence propagation'' to make it more explicit.}
} 

This paper takes the first step to combine viral marketing (influence maximization) with a framework of item adoption grounded in the economic principle of item utility. We propose a \chgins{novel and 
 powerful framework} for capturing the interaction between these two paradigms,  
and \chgins{study the} \emph{social welfare maximization} in this context, i.e., maximize the sum of utilities of itemsets adopted by users at the end of a campaign, in expectation. 
The utility of an itemset is defined to be the valuation of the itemset minus the 
	price of the itemset.
\chgins{Social welfare is well studied in combinatorial auctions, but it has not been well studied in the context of
network propagation and viral marketing.
}

\eat{ 	
\weic{I put ``novel'' in front of the framework, instead of in front of the social welfare maximization. 
The social welfare maximization problem is not that new, but the framework combining viral marketing and utility-based
adoption is new.	
}
} 

\eat{ 	
\weic{For the last sentence in the above paragraph, it is used to be at the beginning of the following paragraph
	with much longer discussion. I downplayed our study as the first one on social welfare in networks.}
} 
	
 
In this paper, we focus on a setting where the items are mutually complementary, by modeling user valuation for itemsets as a \emph{supermodular} function (definition in \textsection\ref{sec:related}). 
Supermodularity captures the intuition that between complementary items, the marginal value-gain of an item w.r.t. a set of items increases as the set grows. Many companies offer complementary products, e.g., Apple
	offers iPhone, and AirPod. The marginal value-gain of AirPod is higher for a user who has bought an iPhone, compared to a user who hasn't. 
\chgins{Complementary items have been well studied in the economics literature and supermodular function is a typical way
	for modeling their valuations (e.g., see \cite{topkis2011supermodularity,Carbaugh16}).} 
\laksRev{As a preview, our experiments show that complementary items are natural and that their valuation is indeed supermodular (Section~\ref{sec:real_data}).} 
We study adoptions of complementary items, by combining a basic stochastic diffusion model with the utility model for item adoption.

In practice, prices of items may be known, but our knowledge of users'  valuation for items may be uncertain. 
Thus, we further add a random noise to the utility function.
\eat{ 
We \emph{do not make any distributional assumptions about the noise}, except for the standard assumption that it is zero mean. 
} 
%
We formulate the optimization problem of finding the optimal allocation of items 
	to seed nodes \chgins{under item budget constraints} so as to maximize the expected social welfare. 
The task is NP-hard, but more challenging is our result that the expected social welfare
	is neither submodular nor supermodular under the reasonable assumption that price and noise are additive.
We show that  
	we can still design an efficient algorithm that achieves a $(1-1/e-\epsilon)$-approximation to the optimal expected social welfare, for any small $\epsilon>0$. 
\chgins{
While our main algorithm is still based on the greedy approach for solving submodular function maximization, its
	analysis is far from trivial, because the objective function is neither submodular nor supermodular.
As part of our proof strategy, we develop a novel {\em block accounting} method for reasoning about expected 
	social welfare for properly defined blocks of items.}

\chgins{
An important feature of our algorithm is that it does not require the valuations or prices of items as the input, and
	merely the fact that item valuation is supermodular while price and noise are additive
	is sufficient
	to guarantee the approximation ratio.
This means that we do not need to obtain the valuations or marginal valuations of items, which may not be \laksRev{straightforward} to
	get in practice.}

\chgins{
To summarize, in this paper, we study the problem of optimal allocation of items to seeds subject to item budgets, 
	such that after network propagation the expected social welfare is maximized, and
	we make the following contributions:}

1. We incorporate utility-based item adoption with influence diffusion into a novel 
	multi-item diffusion model, \emph{Utility-driven IC} (\model) model. \laksRev{\model can support any mix of competing and complementary items.} In this paper, we study the social welfare maximization problem for
	mutually complementary items 
(\textsection\ref{sec:model}). 
	
2. We propose a greedy allocation algorithm, and  show
	that the algorithm achieves a $(1-1/e-\epsilon)$-approximation ratio, even though 
	the social welfare function is neither submodular nor supermodular (\textsection\ref{sec:prop} and \textsection\ref{sec:algo}). 
	\chgins{
	Our main technical contribution is the block accounting method, which distributes social welfare to properly
	defined item blocks.
	The analysis is highly nontrivial and may be of independent interest to other studies.}
	
3. We design a \emph{prefix-preserving} seed selection algorithm for multi-item IM that may be
	of independent interest, with
	running time and memory usage in the same order as the scalable approximation  algorithm IMM \cite{tang15} on the maximum budgeted item,
	regardless of the number of items 
(\textsection\ref{sec:algo}). 

4. We conduct detailed experiments comparing the performance of our algorithm with baselines on five large real networks,
 with both real and synthetic utility configurations. Our results show that our algorithm significantly dominates the baselines in terms of running time or expected social welfare or both 
(\textsection\ref{sec:exp}).  
\section{Background \& Related Work}\label{sec:related}
\subsection{Single Item Influence Maximization}

A social network is represented as a directed graph $G=(V,E,p)$ , $V$ being the set of nodes (users), $E$ the set of edges (connections), with $|V| = n$ and $|E| = m$. The function $p: E \to [0,1]$ specifies influence probabilities (or weights) between users. 
Two of the classic 
diffusion models are independent cascade (IC) and linear threshold (LT). 

We briefly review the IC model. Given a set $S\subset V$ of seeds, diffusion proceeds in discrete time steps. At $t=0$, only the seeds are active. At every time $t>0$, each node $u$ that became active at time $t-1$ makes one attempt at activating each of its inactive out-neighbors $v$, i.e., it tests if the edge $(u,v)$ is ``live'' or ``blocked''. The attempt succeeds (the edge $(u,v)$ is live) with probability $p_{uv} := p(u,v)$. The diffusion stops when no more nodes become active. 

We refer the reader to \cite{kempe03,infbook} for details of these models and their generalizations. 
The {\em influence spread} of a seed set $S$, denoted 
	$\sigma(S)$,  is 
the expected number of active nodes after the diffusion that starts from the seed set $S$ ends.

{\em Influence maximization} (IM) is the problem of finding, for a given number $k$ and a diffusion model, a set $S\subset V$ of $k$ seed nodes that generates the maximum influence spread $\sigma(S)$~\cite{kempe03}. 

Most existing studies on IM rely on the corresponding influence spread
	function $\sigma(S)$ being monotone and submodular.  
A set function $f: 2^U \to \mathbb{R}$ 
is \emph{monotone} if $f(S) \leq f(T)$ whenever $S\subseteq T\subseteq U$; \emph{submodular} if for any $S\subseteq T\subseteq U$ and any $x \in U \setminus T$, $f(S\cup\{x\}) - f(S) \geq f(T\cup\{x\}) - f(T)$; $f$ is \emph{supermodular} if the inequality above is reversed; and $f$ is \emph{modular} if it is both submodular and supermodular. 
Under both the IC and LT models, the IM problem is NP-hard \cite{kempe03} 
and computing $\sigma(S)$ exactly for any $S\subseteq V$ is \#P-hard~\cite{ChenWW10, ChenWW10b}. 
Since $\sigma(\cdot)$ is {\em monotone} and {\em submodular} for both IC and LT, a simple greedy seed selection algorithm together with Monte Carlo simulation for estimating the spread, achieves an $(1-1/e-\epsilon)$-approximation, for any $\epsilon > 0$~\cite{kempe03, submodular,kapraov-etal-greedy-opt-soda-2013}. 
While several heuristics for IM were proposed over the years \cite{ChenWW10,ChenWW10b,ChenWW10c,jung2012,kim2013}, they do not offer any guarantee on the influence spread achieved. 
Borgs et al. \cite{borgs14} proposed the notion of random reverse reachable sets (rr-sets) for spread estimation as an alternative to using MC simulations and paved the way for efficient approximation algorithms for IM. 
Tang et al.~\cite{tang14,tang15} leveraged rr-sets to propose scalable  approximation algorithms for IM called TIM and IMM, 
which are orders of magnitude faster than the classic greedy algorithm making use of MC simulations for estimating the spread \cite{kempe03}. 
Building on the notion of rr-sets, 
a family of scalable approximation algorithms such as TIM, IMM, and SSA, have been developed for IM ~\cite{tang14,tang15,Nguyen2016,Huang2017}.

Motivated by designing an influence oracle, that responds to queries to find seeds for any given budget, Cohen et al. \cite{cohen14} proposed an IM algorithm called SKIM that leverages bottom-$k$ sketches. A noteworthy property of SKIM is that it produces an ordering of the nodes such that any prefix of the ordering consisting of $k$ nodes is guaranteed to have a spread that is at least $(1-1/e-\epsilon)$ times the optimal spread for a seed budget of $k$. Thus, SKIM is essentially a  prefix-preserving algorithm in context of single item IM. However, as shown in \cite{cohen14}, SKIM does not dominate TIM in performance. Given that IMM is orders of magnitude faster than TIM, there is a natural motivation to build a prefix-preserving IM algorithm by adapting IMM to a multi-item context.

Influence maximization under non-submodular models has been studied in previous work~\cite{CLLR15,lu2015,ST17,LCSZ17}.
Most of them show hardness of approximation results~\cite{CLLR15,ST17,LCSZ17}.
In terms of approximation algorithms, Chen et al. rely on a low-rank assumption to
	provide an algorithm solving the non-submodular amphibious influence maximization problem
	with an approximation ratio of $(1-1/e-\epsilon)^3$~\cite{CLLR15}.
Lu et al. use the sandwich approximation to give a problem instance dependent approximation 
	ratio~\cite{lu2015}.
Schoenebeck and Tao provide a dynamic programming algorithm for influence maximization in
	the restricted one-way hierarchical blockmodel~\cite{ST17}.
Li et al. provide an approximation algorithm with approximation ratio $(1-\epsilon)^\ell(1-1/e)$,
	in a network when at most $\ell$ nodes are $\epsilon$-almost submodular and the rest of the nodes
	are submodular~\cite{LCSZ17}.
In contrast, our algorithm in this paper achieves the $(1-1/e-\epsilon)$ approximation ratio (same as
	the ratio for submodular maximization) for a non-submodular objective function, under
	a general network without further assumptions.

\subsection{Multi-item Influence Maximization}

More recently, multiple items have been considered in the context of viral marketing of non-competing items~\cite{dattaMS10,narayanam2012viral}. However their proposed solutions do not provide the typical $(1-1/e-\epsilon)$-approximation guarantee. Specifically
Datta et al.~\cite{dattaMS10} studied IM where propagations of items are assumed to be independent and provided a $1/3$-approximate algorithm. In~\cite{narayanam2012viral}, Narayanam et al. propose an extension of the LT model, where items are partitioned into two sets. A product can be adopted by a node only when it has already adopted a corresponding product in the other set. Such partition of itemsets, with strong dependencies on mutual adoptions of items in the two sets, represents a restricted special case of item adoptions in the real world. 
 
Competitive influence maximization is studied in~\cite{ChenNegOpi11,HeSCJ12,BudakAA11,PathakBS10,borodin10,BharathiKS07,CarnesNWZ07, lu2013} (see \cite{infbook} for a survey), 
where a user adopts at most one item from the set of items being propagated. The works mainly focus on the ``follower's perspective''~\cite{BharathiKS07,CarnesNWZ07, HeSCJ12,BudakAA11}, i.e., given competitor's seed placement, 
select seeds so as to maximize one's own spread, or minimize the competitor's spread. 
Lu et al. \cite{lu2013} 
focused on maximizing the total influence spread of all campaigners from a network host perspective, while ensuring fair allocation.

Lu et al.~\cite{lu2015} introduced a model called \comic capturing both competition and complementarity between a pair of items, leveraging the notion of a node level automaton (NLA). An NLA is a stochastic decision-making automaton governed by transition probabilities for a user adopting an item given what it has already adopted.   
Their model subsumes perfect complementarity and pure competition as special cases. 
However, their main study is confined to the diffusion of two items, and a straightforward extension to multiple items would need an exponential number of parameters in the number of items.
Moreover, their general parameter settings could lead to anomalies such as   one item complementing a second item but the second one competing with the first one, or being indifferent to it.

{\sl All of the above works on multiple item propagations focus on maximizing the expected number of item adoptions, which is not aligned with social welfare.}  

Myers and Leskovec analyzed the effects of different cascades on users and 
predicted the likelihood that a user will adopt an item, seeing the cascades in which the user participated~\cite{myers12}.
McAuley et al.\cite{mcauley15} proposed a method 
to learn complementary relationships between products from user reviews.
None of the works models the diffusion of complementary items, nor study the IM problem in this context.

\subsection{Combinatorial Auctions} 

Combinatorial auctions are widely studied and a survey is beyond the scope of this paper. Instead, we discuss a few key papers. In economics, adoption of items by users is modeled in terms of the utility that the user derives from the adoption~\cite{hirshleifer1978, myerson1981optimal, rasmusen1994, nisan2007}. A classic problem 
is given $m$ users and $n$ items and the utility function of users for various subsets of items, find an allocation of items to users such that the social welfare, i.e., the sum of utilities of users resulting from the allocation, is maximized. This intractable problem has been studied in both offline and online settings \cite{cramton-etal-ca-book-2007, feige-vondrak-demand-2010, kapraov-etal-greedy-opt-soda-2013, korula-etal-online-swm-arxiv-2017} and various approximation algorithms have been developed. 
All of them assume access to a value oracle or a demand oracle. A value oracle is a black box, which given a set of items as a query, returns the value of the itemset. A demand oracle is a black box, which given an assignment of prices to items, returns the itemset with maximum utility, i.e., value minus price.  
Also, the utility function in these settings is typically assumed to be sub-additive and as a result, this property extends to social welfare. Notably, these works do not consider the interaction of utility-maximizing item adoption with recursive propagation through a network. On the other hand, they consider more general settings where the utility functions are user-specific.

Inspired by the economics literature, we base item adoptions on item utility. Specifically, items have a price and a valuation and the difference is the utility. It is a well-accepted principle in economics and auction theory \cite{cramton-etal-ca-book-2007,snyder08} that users (agents), presented with a set of items, adopt a subset of items 
that maximizes their utility. 
It is this principle that we use in our framework to govern which users adopt what items.

The use of utility naturally leads to the notion of \emph{social welfare} and we study the problem of assigning seed nodes to various items in order to maximize expected social welfare, in a setting where items are complementary. 
To our knowledge, in the context of viral marketing, we are the first to study the problem of maximizing (expected) social welfare. 

\subsection{Welfare maximization on social networks}

There are a few studies related to welfare maximization on social networks, but they all have significant differences with
	our model and problem setting.
Sun et al. \cite{SunCLWSZL11} study participation maximization in the context of online discussion forums. 
An item in that context is a discussion topic, and adopting an item means posting or replying on the topic.
Item adoptions do propagate in the network, but 
	(a) item propagations are independent (i.e., valuation of itemsets is additive rather than supermodular or
	submodular), and
	(b) they have a budget on the number of items each seed node can be allocated with, rather than 
	on the number of seeds each item can be allocated to as studied in our model.
Bhattacharya et al. \cite{BhattacharyaDHS17} consider item allocations to nodes for welfare maximization in a network
	with network externalities, but the major differences with our  problem are:
	(a) they use network externalities to model social influence, i.e., a user's valuation of an item is affected
	by the number of her one- or two-hop neighbors in the network adopting the same item, but network externalities do
	not model the \emph{propagation} of influence and item adoptions, our main focus in modeling the viral marketing
	effect; 
	(b) they consider unit demand or bounded demand on each node, which means items are competing against one another
	on every node, while our study focuses on the case of complementary items rather than competing items, and item
	bundling is a key component in our solution;
	(c) they do not have budget on items so 
	an item could be allocated to any number of nodes, while
	we have a budget on the number of nodes that can be allocated to an item as seeds and we rely on propagation for
	more nodes to adopt items.
Despite these major differences, we will do an empirical comparison of our algorithm versus their algorithms to demonstrate
	that with propagation we can achieve the same social welfare with only a fraction of item budgets used in their
	solution. 
Abramowitz and Anshelevich~\cite{AbramowitzA18} study network formation with various constraints to maximize social
	welfare, but it has no item allocation, no item complementarity, and no influence propagation, and thus is further
	away from our work.
In summary, to our knowledge, our study is the only one addressing social welfare maximization in a network with 
	influence propagation, complementary items, and budget limits on items.

\section{\model Model}\label{sec:model}
\begin{table*}[t!]
	\scriptsize
	\centering
	\hspace*{-1mm}
  \begin{tabular}{|c|c|c|c|} \hline
  $G, V, E, n$ and $m$ & Graph, node set, edge set, number of nodes and number of edges  & $p : E \rightarrow [0,1]$ & Influence weight function  \\ \hline
  $\allitems$ & Universe of items & $\price$, $\val$, $\noise$ and $\util$ & Price, Value, Noise and Utility \\ \hline
  $\bvec$ & Budget vector & $\bmax$ & Maximum budget  \\ \hline 
  $\allalloc$ & Seed allocation, i.e. set of node-item pairs & $S$ & Seed nodes  \\ \hline
  $S_i^{\allalloc}$ & Seed nodes of item $i$ in allocation $\allalloc$ & $S^{\allalloc}$ & All seed nodes of allocation $\allalloc$  \\ \hline 
  ${\bf I}_v^{\allalloc}$ & Items allocated to seed node $v$ in allocation $\allalloc$ & $\awares(\user,t)$ and $\adopts(\user,t)$ & Desire and adoption set of $\user$ at time $t$ in allocation $\allalloc$   \\ \hline
  $\sigma$ and $\rho$ & Expected adoption and social welfare & $W$, $W^E$, $W^N$ & Possible world, edge and noise possible world   \\ \hline
  $Grd$ and $OPT$ & Greedy and optimal allocation & $B$ and $\allblocks$ & A block and a sequence of item disjoint blocks  \\ \hline
  $\bbude_i$ & Effective budget of block $B_i$ & $B^a_i$ and $a_i$ & Anchor block and anchor item of block $B_i$ \\ \hline
  \end{tabular}
  	\caption{Table of notations}
	 \label{tab:notations}
\end{table*}

In this section, we propose a novel model called \emph{utility driven independent cascade} model (\model for short) that combines the diffusion dynamics of the classic IC model with an item adoption framework where decisions are governed by utility. Table \ref{tab:notations} summarizes the notations used henceforth.
 
\subsection{Utility based adoption}

Utility is a widely studied concept in economics and is used to model item adoption decisions of users \cite{myerson1981optimal, feige-vondrak-demand-2010, nisan2007}. 
We next briefly review utility and provide the  specific formulation we use in this paper.
For general definitions related to utility, the reader is referred to \cite{nisan2007,feige-vondrak-demand-2010}

We let $\allitems$ denote a finite universe of items. The utility of a set of items $\itemset \subseteq \allitems$ for a user is the pay-off of $I$ to the user and 
	depends on the aggregate effect of three components: the price $\price$ that the user needs to pay, the valuation $\val$ that the user has for $\itemset$ and a random noise term $\noise$, used to model the uncertainty in our knowledge of the user's valuation on items, where $\price$, $\val$ and $\noise$ are all set functions over items.
For an item $i \in \allitems$, $\price(i) > 0$ denotes its price. We assume that price is additive, i.e., for an itemset $\itemset \subseteq \allitems$, $\price(I) = \sum_{i \in \itemset} \price(i)$. 
Notice that \model can handle any generic valuation function. In \textsection \ref{sec:complementary} we focus on complementary products. Hence we assume that $\val$ is supermodular (definition in \textsection \ref{sec:related}),
	meaning that the marginal value of an item with respect to an itemset $I$ 
	increases as $I$ grows. 
We also assume $\val$ is monotone since it is a natural property for valuations.  
For $i \in \allitems$, $\noise(i) \sim \mathcal{D}_i$ denotes the noise term associated with item $i$, 
where the noise may be drawn from any distribution $\mathcal{D}_i$ having a zero mean. Every item has an independent noise distribution. For a set of items $I \subseteq \allitems$, we assume the noise is additive, i.e., the noise of $I$, $\noise(I) := \sum_{i \in I} \noise(i)$. 
\chgins{Similar assumptions on additive noise are used in economics theory \cite{hirshleifer1978,Albertbidding07}.}

Finally, the utility of an itemset $\itemset$ is $\util(\itemset) = \val(\itemset) - \price(\itemset) + \noise(\itemset)$. Since noise is a random variable, utility is also random. 
Since noise is drawn from a zero mean distribution, $\mathbb{E}[\util(I)] = \val(\itemset) - \price(\itemset)$. 
We assume $\val(\emptyset) = 0$.

\subsection{Diffusion Dynamics}\label{sec:decomic} 

\subsubsection{Seed allocation}

Let $\vec{b} = (b_1, ..., b_{|\allitems|})$ be a vector of natural numbers representing the budgets associated with the items. An item's budget specifies the number of seed nodes that may be assigned to that item. We sometimes abuse notation and write $b_i \in \bvec$ to indicate that $b_i$ is one of the item budgets. We denote the maximum budget as $\bmax := max\{b_i\mid b_i\in \bvec\}$. We define an \emph{allocation} as a relation $\allalloc \subset V \times \allitems$ such that $\forall i\in\allitems: |\{(v,i)\mid v\in V\}| \le b_i$. In words, each item is assigned a set of nodes whose size is under the item's budget. We refer to the nodes $S_i^{\allalloc} := \{v \mid (v,i) \in \allalloc\}$ as the \emph{seed nodes} of $\allalloc$ for item $i$ and to the nodes $S^{\allalloc} := \bigcup_{i\in\allitems} S_i^{\allalloc}$ as the \emph{seed nodes} of $\allalloc$. We denote the set of items allocated to a node $v\in V$ as ${\bf I}_v^{\allalloc} := \{i\in\allitems \mid (v,i) \in \allalloc\}$. 

\subsubsection{Desire and adoption}

Every node maintains two sets of items -- desire set and adoption set. Desire set is the set of items that the node has been informed about (and thus potentially desires), via propagation or seeding. Adoption set is the subset of the desire set that the node adopts.  
At any time a node selects, from its desire set at that time, the subset of items that maximizes the utility, and adopts it. If there is a tie in the maximum utility between itemsets, then it is broken in favor of larger itemsets. We later show in Lemma \ref{lem:itemunion} of \textsection \ref{sec:prop} that breaking ties in this way results in a well-defined adoption behavior of the nodes. 
Following previous literature, we consider a progressive model: once a node desires an item, it remains in the node's desire set forever; similarly, once an item is adopted by a node, it cannot be unadopted later. 

For a node $u$, $\aware^{\allalloc}(u,t)$ denotes its desire set and $\adopt^{\allalloc}(u,t)$ denotes its adoption set at time $t$, pertinent to an allocation $\allalloc$. We omit the time argument $t$ to refer to the adoption (desire) set at the end of diffusion. 

We now present the diffusion process of \model.

\subsubsection{The diffusion model}

In the beginning of any diffusion, the noise terms of all items are sampled,  which are then used till the diffusion terminates. The diffusion then proceeds in discrete time steps, starting from $t=1$. 
Given an allocation $\allalloc$ at $t=1$, the seed nodes have their desire sets initialized : $\forall v \in S^{\allalloc}$, $\awares(v,1) = {\bf  I}_v^{\allalloc}$. 
Seed nodes then adopt the subset of items from the desire set that 
	maximizes the utility, breaking ties if needed 
	in favor of sets of larger cardinality. 
Thus, a seed node may adopt just a subset of items allocated to it.

Once a seed node $u'$ adopts an item $i$, it influences its out-neighbor $u$ 
	with probability $p_{u', u}$, and if
	it succeeds, then $i$ is added to the desire set of $u$ at time $t =2$. The rest of the diffusion process is described in Fig. \ref{fig:ecomicmodel}.

\begin{figure}[h]
\begin{framed}
{

    \begin{description}[style=unboxed,leftmargin=8pt]

\item[1.]
\textbf{Edge transition.} At every time step $t > 1$, for a node $u'$ that has adopted at least one new item at $t-1$, its outgoing edges are tested for transition.
For an untested edge $(u',\ua)$, flip a biased coin independently: $(u',\ua)$ is {\em live} w.p.\ $p_{u',\ua}$ and {\em blocked} w.p.\ $1-p_{u',\ua}$. Each edge is tested {\em at most once} in the entire diffusion process and its status is remembered for the duration of a diffusion process. 


Then for each node $\user$ that has at least one in-neighbor $u'$ (with a live edge $(u', \user)$) which adopted at least one item at $t-1$, $\ua$ is tested for possible item adoption (2-3 below).


\item[2.]
\textbf{Generating desire Set.}
The desire set of node $\user$ at time $t$, \\
$\awares(\user,t) = \awares(\user, t-1)   \cup_{u' \in N^-(\user)} (\adopts(u',t-1) )$, where $N^-(\user) = \{u' \mid (u', \user) \mbox{ is live}\}$ denotes the set of in-neighbors of $\user$ having a live edge connecting to $\user$. 

\item[3.]
\textbf{Node adoption.}
Node $\user$ determines the utilities for all subsets of items of the desire set  $\awares(\user,t)$. $\user$ then adopts a set $T^* \subseteq \awares(\user,t)$ such that  $T^* = \argmax_{T \in 2^{\awares(\user,t)}} \{\util(T) \mid T \supseteq \adopts(\user,t-1)\;\wedge\;\util(T)\ge 0\}$. $\adopts(\user, t)$ is set to $T^*$. 

    \end{description}

}
\end{framed}
\caption{Diffusion dynamics under \model model}
\label{fig:ecomicmodel}
\end{figure} 


\chgins{
We illustrate the diffusion under \model using an example shown in Figure \ref{fig:uic_diffusion}. The graph $G$ with edge probabilities and the utilities of the two items after sampling the noise terms, are shown on the left side. At time $t=1$, node $v_1$ is seeded with item $i_1$ and $v_3$ with $i_2$, hence they desire those items respectively. Since $i_1$ (resp. $i_2$) has a positive (resp. negative) individual utility, $v_1$ adopts $i_1$ (resp. $v_3$ does not adopt $i_2$). However 
$i_2$ remains in the desire set of $v_3$. Then at $t=2$, outgoing edges of $v_1$ are tested for transition: edge $(v_1,v_3)$ fails (shown as red dotted line), but edge $(v_1,v_2)$ succeeds (green solid line). 
Consequently $v_2$ desires and adopts $i_1$. Next at $t=3$, $v_2$'s outgoing edge $(v_2,v_3)$ is tested. As it succeeds, $v_3$ desires $i_1$. Since it already had $i_2$ in its desire set, it adopts the set $\{i_1,i_2\}$. Since there is no outgoing edge from $v_3$, the propagation ends. 
} 

\begin{figure*}
\begin{minipage}{.32\textwidth}
  \hspace{-35mm}
  \includegraphics[width=.75\linewidth]{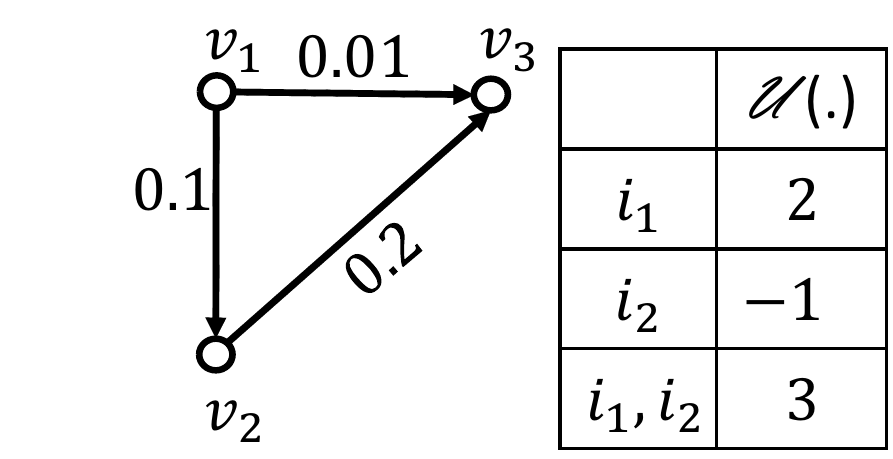}
\end{minipage}%
\begin{minipage}{.35\textwidth}
  \hspace{-40mm}
  \includegraphics[width=2.5\linewidth]{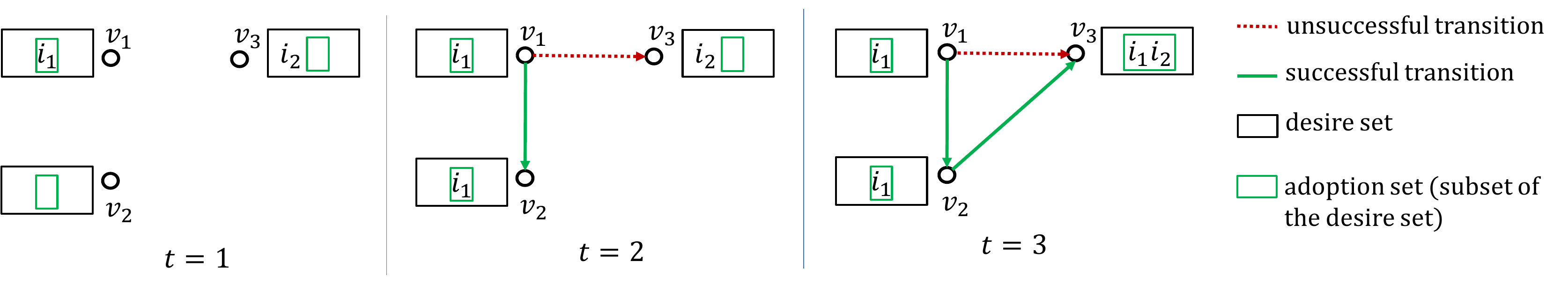}
\end{minipage}
 \captionof{figure}{Illustrating propagation of items under \model model; for simplicity, we assume noise is zero.}
  \label{fig:uic_diffusion}
\end{figure*} 

\subsection{Definition of Social welfare Maximization}\label{sec:sw} 

Let $G = (V,E,p)$ be a social network, $\allitems$ the universe of items under consideration. 
Here, we consider a novel utility-based objective called {\em social welfare}, which is the sum of all users' utilities of itemsets adopted by them after propagation converges. Formally, $\E[\util(\adopt^{\allalloc}(u))]$ is the expected utility that a user $\user$ enjoys for a seed allocation $\allalloc$ after propagation ends. Then the \emph{expected social welfare} (also known as ``consumer surplus'' in 
algorithmic game theory) for $\allalloc$,  is $\rho(\allalloc) = \sum_{u \in V}\E[\util(\adopt^{\allalloc}(\user))]$, where the expectation is over both the randomness of propagation and randomness of noise. 

\spara{Key features of \model} 

Our utility driven model has several benefits over existing models. 

Firstly, the seed users in our model are treated as rational users. Thus they also go through the same utility based decision making like every other user of the network.

Secondly, \comic cannot handle the arrival of a set of items together. It had to use arbitrary tie-breaking in a case when a node becomes aware more than one items simultaneously, to put an order in the adoption. In \model, we treat this by creating an explicit desire set for nodes first. The utility is a set function as opposed to the point probability of GAP. Therefore even if more than one items arrive at the same time instance, the utility can treat them as a set, without the need of enforcing an explicit order.

Third, the notion of utility opens up a whole new objective of social welfare maximization, where instead of maximizing just the adoption, the utility earned from the adoptions is aimed to be maximized. No other model reported to date, has been able to study social welfare maximization. \model is the first framework which enables the study of utility in IM context.

Fourth, for complimentary products under \model, a greedy allocation algorithm that preserves $1 - \frac{1}{e}$ approximation guaranty with respect to optimal social welfare, although social welfare is not submodular in seed size. This greedy algorithm is independent of the model parameters. Hence it can be easily extended to multiple items, whereas for \comic extending the algorithm beyond two items was difficult due to parameter explosion.

We define the problem of maximizing expected social welfare (\emph{WelMax}) as follows. We refer to $\val, \price, \noise$,  as the model parameters and denote them collectively as $\parameterset$. 

\begin{problem} \label{prob:welmax}
[WelMax]
Given $G = (V,E,p)$, the set of model parameters $\parameterset$, and budget vector $\bvec$, find a seed allocation $\allalloc^*$, such that $\forall i \in \allitems$, $|\alliseeds^{\allalloc^*}| \leq b_i$ and $\allalloc^*$ maximizes the expected social welfare, i.e., $\allalloc^* = \argmax_{\allalloc} \rho(\allalloc)$.
\end{problem}

Unfortunately, WelMax is NP-hard. 

\begin{proposition}
WelMax in the UIC model is NP-hard.
\end{proposition}

\begin{proof}
It is easy to verify that 
Influence maximization under the IC model, an
	NP hard problem, is a special case
	of WelMax.

The result follows from the fact that the IM problem under the IC model is a special case of WelMax: let $\allitems = \{i\}$, set $\val(i) = 1$, $\price(i) = 0$ and set the noise term for item $i$ to $0$. This makes $\util(i) = 1$ so any influenced node will adopt $i$. Thus, the expected social welfare is simply the expected spread. We know maximizing expected spread under the IC model is NP-hard \cite{kempe03}.
\end{proof} 

\subsubsection{Function Types} 

Notice that the functions $\val$ and $\util$ are functions over sets of items, whereas $\sigma$ is a function over sets of network nodes, and $\rho$ is a function over allocations, which are sets of (node, item) pairs. When we speak of a certain property (e.g., submodularity) of a function of a given type, the property is meant w.r.t. the applicable type. E.g., $\sigma$ is 
monotone and submodular w.r.t. sets of nodes. 

\subsubsection{Design choices}

In the UIC model, the desire set of a user is triggered either by seeding or by the influence on a user as her peers adopt items. Thus following standard practice in IM models, we keep it progressive: a desire set never shrinks. On the other hand, the adoption decisions are driven by a standard assumption in economics  \cite{boadway1984welfare}, that users aim to maximize the utility when they adopt item(sets). \model inherits this assumption to govern adoption decisions of the users. In \model, we assume price is additive. There are different ways of pricing a bundle of items: additivity is a simple and natural pricing model in the absence of discounts \cite{chang1979pricing}. Further, we use supermodular value functions to model the effect of complementarity which follows the standard practice in the economics literature \cite{topkis2011supermodularity, milgrom1995complementarities}. 
\chgins{Finally, our way of modeling the noise can be viewed as reflecting the uncertainty in
	the population's reaction to an item.
One may further introduce personalized noise to model individual uncertainty,
	but this would make algorithm design and analysis more
	difficult. Our approximation bound would not hold when noise is personalized and when valuation is not supermodular.}
Although we make specific design choices in this work for simplicity and tractability of the model, the 
\model model can encompass any general form of value, price, and noise parameters and works for any triggering model \cite{kempe03}.


\section{\model For Complementary Products } \label{sec:complementary}

In this section, we focus on a setting where the items are mutually complementary, by modeling user valuation for itemsets as a \emph{supermodular} function. Recall that a function $f: 2^U \rightarrow \mathbb{R}$ is supermodular if for any subsets $S \subset T \subset U$ and item $x \in U \setminus T$, $f(S\cup\{x\}) - f(S) \le f(T\cup\{x\}) - f(T)$. 
Supermodularity captures the intuition that between complementary items, the marginal value-gain of an item w.r.t. a set of items increases as the set grows. Many companies offer complementary products, e.g., Apple
	offers iPhone, and AirPod. The marginal value-gain of AirPod is higher for a user who has bought an iPhone, compared to a user who hasn't. 
\chgins{Complementary items have been well studied in the economics literature and supermodular function is a typical way
	for modeling their valuations (e.g., see \cite{topkis2011supermodularity,Carbaugh16}).} 
\laksRev{As a preview, our experiments show that complementary items are natural and that their valuation is indeed supermodular.} 
We study adoptions of complementary items, by combining a basic stochastic diffusion model with the utility model for item adoption. The highlights of the section are as follows:

1. We propose a greedy allocation algorithm, and  show
	that the algorithm achieves a $(1-1/e-\epsilon)$-approximation ratio, even though 
	the social welfare function is neither submodular nor supermodular (\textsection\ref{sec:prop} and \textsection\ref{sec:algo}). 
	\chgins{
	Our main technical contribution is the block accounting method, which distributes social welfare to properly
	defined item blocks.
	The analysis is highly nontrivial and may be of independent interest to other studies.}
	
2. We design a \emph{prefix-preserving} seed selection algorithm for multi-item IM that may be
	of independent interest, with
	running time and memory usage in the same order as the scalable approximation  algorithm IMM \cite{tang15} on the maximum budgeted item,
	regardless of the number of items 
(\textsection\ref{sec:algo}). 

3. We conduct detailed experiments comparing the performance of our algorithm with baselines on five large real networks,
 with both real and synthetic utility configurations. Our results show that our algorithm significantly dominates the baselines in terms of running time or expected social welfare or both 
(\textsection\ref{sec:exp}).

\subsection{Properties Of \model Under Supermodular Valuations}\label{sec:prop}

Since WelMax is NP-hard, we explore properties of the welfare function --  monotonicity and submodularity, which can help us design efficient approximation strategies. We begin with an equivalent possible world model to help our analysis. 

\subsubsection{Possible world model}\label{sec:prop_pw}

Given an instance $\langle G, {\sf Param}\rangle$ of \model, where $G=(V,E,p)$, we define a \emph{possible world} associated with the instance, as a pair $W = (W^E, W^N)$, where $W^E$ is an  {\em edge possible world} (edge world), and $W^N$ is a {\em noise possible world} (noise world); $W^E$ is a sample graph drawn from the distribution associated with $G$ by sampling edges, and $W^N$ is a sample of noise terms for each item in $\allitems$, drawn from the corresponding item's noise distribution in {\sf Param}. 
 
As all the random terms are sampled, propagation and adoption in $W$ is fully deterministic. For nodes $u, v \in V$, we say $v$ is reachable from $u$ in $W$ if there is a directed path from $u$ to $v$ in the deterministic graph $W^E$. $\noise_W(i)$ denotes the sampled noise for item $i$ and $\util_W(I)$ denotes the (deterministic) utility of itemset $I$, in world $W$. For a node $u$ and an allocation $\allalloc$, we denote its desire and adoption sets at time $t$ in world $W$ as $\awarews(u,t)$ and $\adoptws(u,t)$ respectively. When only the noise terms are sampled, i.e., in a noise world $W^N$, the utilities are deterministic, but the propagation remains random. 

Given a possible world $W = (W^E, W^N)$ and an allocation $\allalloc$, a node $v \in V$ adopts a set of items as follows: (i) if $v$ is a seed node, then it desires ${\bf I}_v^{\allalloc}$ at time $t=1$ and adopts an itemset $\adoptws(v,1) := \argmax \{\util_W(I) \mid I\subseteq {\bf I}_v^{\allalloc}\}$; (ii) if $v$ is a non-seed node, and $t>1$, then it desires the itemset $\awarews(v,t) := (\bigcup_{u\in N^{-1}_W(v)}\adoptws(u,t-1)) \cup \awarews(v,t-1)$,
where $N^{-1}_W(v)$ denotes the in-neighbors of $v$ in the deterministic graph $W^E$, i.e., at time $t$, node $v$ desires items that it desired at $(t-1)$ as well as items any of its in-neighbors in $W^E$ adopted at $(t-1)$; node $v$  then adopts the itemset $\adoptws(v,t) := \argmax \{\util_W(I) \mid I \subseteq \awarews(v,t) \;\&\; \adoptws(v,t-1)\subseteq I\}$. If there is more than one itemset in $\awarews(v,t)$ with the same maximum utility, we assume that $v$ breaks ties in favor of the set with the larger cardinality.

$\val(\cdot)$ is supermodular while $\price(\cdot)$ and $\noise_W(\cdot)$ are additive and hence modular, so it immediately follows that $\util_W(\cdot)$ is supermodular with respect to sets of items.
Thus the expectation of utility w.r.t. edge worlds is supermodular. However, $\utilow$ is not monotone, because adding an item with a very high price may decrease the utility. 

We will show a basic property, which helps us showing that the adoption behavior of the nodes is well defined in $\model$. In 
any possible world, given a set of items that a node desires, there is a unique set of items that it adopts. Specifically, if there are multiple sets tied for utility, the node will adopt their union. For a set function $f:2^U\rightarrow R$, we define $f(T\mid S) = f(S\cup T) - f(S)$. 

We say that an itemset $A$ is a \emph{local maximum} w.r.t. the utility function $\util_W$, if the utility of $A$ is the maximum among all its subsets, i.e., $\util_W(A) = \max_{A'\subseteq A} \util_W(A')$. 
The following lemma is based on simple 
	algebraic manipulations on the definitions of supermodularity and local maximum.

\begin{restatable}{lemma}{lemitemunion}
	\label{lem:itemunion}
(Local maximum). Let $W$ be a possible world and $A, B \subseteq \allitems$ be any itemsets 
such that $A$ and $B$ are local maximum with respect to $\util_W$.
Then $(A\cup B)$ is also a local maximum with respect to $\util_W$, 
i.e., $\util_W(A\cup B) = \max_{C\subseteq A\cup B} \util_W(C)$.
\end{restatable}

\begin{proof}
For any subset $C \subseteq A\cup B$, we have
	\begin{align} 
	\util_W(C) & = \util_W(C\setminus B \mid B \cap C) + \util_W(B \cap C)  \nonumber \\
	& \le \util_W(C\setminus B \mid B) + \util_W(B)   \label{eq:CB} \\
	& = \util_W(C \cup B) \nonumber \\
	& = \util_W(B \mid C\setminus B) + \util_W(C \setminus B) \nonumber \\
	& \le \util_W(B \mid A) + \util_W (A)  \label{eq:AB} \\
	& = \util_W(A \cup B). \nonumber 
	\end{align} 
	Inequality~\eqref{eq:CB} follows from applying supermodularity of $\util_W$
	on the first term, and applying local maximum of $B$ on the second term.
	Inequality~\eqref{eq:AB} follows applying supermodularity of $\util_W$
	on the first term, and applying local maximum of $A$ on the second term.
\end{proof}

An immediate consequence of Lemma~\ref{lem:itemunion} is that when two itemsets have the same largest utility, their union must also have the largest utility, and thus our tie-breaking rule is well-defined.
Another consequence is the following lemma.

\begin{restatable}{lemma}{lemlocalmaximum} 
\label{lem:adoptmax}
	For any node $u$ and any time $t$, the itemset adopted by $u$ at time $t$, 
	$\adoptws(u,t)$, must be
	a local maximum.
\end{restatable}

\begin{proof}
We prove by an induction on $t$.
The base case of $t=1$ is true because by the model, node $u$ adopts the local maximum among
	all subsets of items allocated to it.
For the induction step, suppose for a contradiction that $\adoptws(u,t)$ is not a local maximum
	but $\adoptws(u,t-1)$ is a local maximum.
Then there must exist a $C\subset \adoptws(u,t)$ that is a local maximum and
	$\util_W(C) > \util_W(\adoptws(u,t))$.
By Lemma~\ref{lem:itemunion}, $C \cup \adoptws(u,t-1)$ is also a local maximum, and thus
	$C \cup \adoptws(u,t-1)$ cannot be $\adoptws(u,t)$.
But since $\util_W(C \cup \adoptws(u,t-1)) \geq \util_W(C) > \util_W(\adoptws(u,t))$, $u$
	should adopt $C \cup \adoptws(u,t-1)$ instead of $\adoptws(u,t)$, a contradiction.
\end{proof}

Our next result shows that in any given possible world, adoption of items propagates through reachability. 
Reachability is a key property to be used later in Lemmas \ref{lem:greedyasw} and \ref{lem:allasw} while establishing the approximation guarantee of our algorithm.

\begin{restatable}{lemma}{lemreachable}
	\label{lem:reachable}
	(Reachability). For any item $i$ and any possible world $W$, if a node $u$ adopts $i$ under allocation $\allalloc$, then all nodes that are reachable from $u$ in the world $W$ also adopt $i$.
\end{restatable}

\begin{proof}
Consider a possible world $W$ and a node $u$ that adopts item $i$. Consider any node $v$ reachable from $u$ in $W$ that does not adopt $i$. Let $(u, v_1, ..., v_k, v)$ be a path in $W^E$. Assume w.l.o.g. that $v$ is the first node on the path that does not adopt $i$. $\adoptws(v_k,t)$ and $\adoptws(v,t+1)$ respectively are the itemsets adopted by $v_k$ at time $t$ and by $v$ at time $t+1$. 
	Let $J = \adoptws(v_k,t) \cup \adoptws(v,t+1)$. Clearly $i \in J$ and $J \subset \awarews(v,t+1)$, desire set of $v$ at $t+1$.
	We know that both $\adoptws(v_k,t)$ and $\adoptws(v,t+1)$ are local maximums 
		by Lemma~\ref{lem:adoptmax}. 
	Then by Lemma \ref{lem:itemunion}, $J$ is also a local maximum, hence $\utilw(J) \geq \utilw(\adoptws(v,t+1))$, as $\adoptws(v,t+1) \subset J$. Also, $|J|> |\adoptws(v,t+1)|$, as $J$ contains at least one more item $i$. Thus as per our diffusion model $v$ at time $t$ should adopt the larger cardinality set $J$. Hence $i$ is adopted by $v$.
\end{proof}

The \emph{social welfare} of an allocation $\allalloc$ in a possible world $W = (W^E, W^N)$ is defined as the sum of utilities of itemsets adopted by nodes, i.e., $\rho_W(\allalloc) := \sum_{v\in V} \util(\adoptws(v))$. 
The \emph{expected social welfare} of an allocation $\allalloc$ is $\rho(\allalloc) := \mathbb{E}_{W^E}[\mathbb{E}_{W^N}[\rho_W(\allalloc)]] = \mathbb{E}_{W^N}[\mathbb{E}_{W^E}[\rho_W(\allalloc)]]$. 
It is straightforward to show that the expected social welfare of allocation $\allalloc$ defined in \textsection\ref{sec:sw} is equivalent to the above definition. 

We now proceed to investigate the properties of social welfare.

\subsubsection{Properties of social welfare} 

The following theorem summarizes the property of social welfare function. The key intuition is that in each
possible world, the social welfare is monotone, a result proved by induction on the propagation time. However
it is not submodular because the valuation is supermodular, and it is not supermodular because the propagation
	based on IC model would have submodular influence coverage.

\begin{restatable}{theorem}{thmnotsubsuper}
Expected social welfare is monotone with respect to the sets of node-item allocation pairs.
However it is neither submodular nor supermodular.
\end{restatable}

\begin{proof} 

To prove monotonicity, we show by induction on propagation time that the social welfare in any world $W$ is monotone. The result follows upon taking expectation. Consider allocations $\allalloc \subseteq \allalloc'$ and any node $v$.  
	
	\noindent 
	\underline{Base Case}: At $t=1$, desire happens by seeding. By assumption, ${\bf I}_v^{\allalloc} \subseteq {\bf I}_v^{\allalloc'}$. Thus, $\aware_W^{\allalloc}(v,1)\subseteq \aware_W^{\allalloc'}(v,1)$, where $\aware_W^{\allalloc}(v,1)$ denotes the desire set of $v$ in world $W$ under allocation $\allalloc$. Suppose $J := \adopt_W^{\allalloc}(v,1) \setminus  \adopt_W^{\allalloc'}(v,1)$ is non-empty. From the semantics of adoption of itemsets, we have $\util_W(J\mid \adopt_W^{\allalloc}(v,1) \setminus J) \ge 0$. Now, $\adopt_W^{\allalloc}(v,1) \setminus J \subseteq \adopt_W^{\allalloc'}(v,1)$. By supermodularity of utility, $\util_W(J\mid \adopt_W^{\allalloc'}(v,1)) \ge 0$.
	Since $J \subseteq \adopt_W^{\allalloc}(v,1) \subseteq \aware_W^{\allalloc}(v,1)
	\subseteq \aware_W^{\allalloc'}(v,1)$,
	by the semantics of itemset adoption, the set $J \cup \adopt_W^{\allalloc'}(v,1)$ will be adopted by $v$ at time $1$, a contradiction to the assumption that
	$\adopt_W^{\allalloc'}(v,1)$ is the adopted itemset by $v$ at time $1$.
	
	\noindent
	\underline{Induction}: By Lemma~\ref{lem:reachable}, we know that once a node adopts an item, all nodes reachable from it in $W^E$ also adopt that item. Furthermore, reachability is monotone in seed sets. From this, it follows that $\adopt_W^{\allalloc}(v,\tau+1) \subseteq \adopt_W^{\allalloc'}(v,\tau+1)$. Define $\adopt_W^{\allalloc}(v) := \bigcup_t \adopt_W^{\allalloc}(v,t)$. By definition, an adopted itemset has a non-negative utility, so we have $\rho_W(\allalloc) = \sum_{v\in V} \util_W(\adopt_W^{\allalloc}(v))\le \sum_{v\in V} \util_W(\adopt_W^{\allalloc'}(v)) = \rho_W(\allalloc')$. This shows that the social welfare in any possible world is monotone, as was to be shown.
	
	For submodularity and supermodularity, we give counterexamples. 
	Consider a network with single node $u$ and two items $i_1$ and $i_2$. Let $\price(i_1) > \val(i_1)$ and $\price(i_2) > \val(i_2)$. However $\val(\{i_1,  i_2\}) > \price(i_1) + \price(i_2)$. Assume that noise terms are bounded random variables, i.e., $|\noise(i_j)| \le |\val(i_j) - \price(i_j)|$, $j=1,2$. 
	Thus expected individual utility of $i_1$ or $i_2$ is negative, but when they are offered together, the expected utility is positive. Now consider two seed allocations $\allalloc = \emptyset$ and $\allalloc' = \{(u,i_1)\}$. Let the additional allocation pair be $(u,i_2)$. Now 
	$\rho(\allalloc \cup \{(u,i_2)\}) - \rho(\allalloc) = 0 - 0 =0$: for $\allalloc$, no items are adopted and for $\allalloc \cup \{(u, i_2)\}$ the noise $\noise(i_2)$ cannot affect adoption decision in any possible world, so $i_2$ will not be adopted by $u$ in any world. 
	
	However, $\rho(\allalloc' \cup \{(u,i_2)\}) - \rho(\allalloc') > 0$, as under allocation $\allalloc'$, $i_1$ is not adopted by $u$ in any world, while under allocation $\allalloc' \cup \{(u,i_2)\}$, $u$ will adopt $\{i_1, i_2\}$ in every world, resulting in positive social welfare and breaking submodularity.

	For supermodularity, consider a network consisting of two nodes $v_1$ and $v_2$ with a single directed edge from $v_1$ to $v_2$, with probability 1. Let there be one item $i$ whose deterministic utility is positive, i.e., $\val(i) > \price(i)$. Again, assume that the noise term $\noise(i)$ is a bounded random variable, i.e., $|\noise(i)| \le |\val(i) - \price(i)|$. 
	Now consider two seed allocations $\allalloc = \emptyset$ and $\allalloc' = \{(v_1, i)\}$. Let the additional pair be $(v_2, i)$. Under allocation $\allalloc'$, both nodes $v_1$ and $v_2$ will adopt $i$ in every possible world. Hence adding the additional pair $(v_2, i)$ does not change item adoption in any world and consequently the expected social welfare is unchanged. Thus we have,
	
	\begin{align*}
	\rho(\allalloc \cup \{(v_2, i)\}) - \rho(\allalloc) &= \mathbb{E}[\util(i_1)] >0  \\
	&= \rho(\allalloc' \cup \{(v_2, i)\}) - \rho(\allalloc')
	\end{align*} 
	which breaks supermodularity. 
\end{proof}

The node level adoption exhibits supermodularity because the utility function is supermodular, but the propagation behavior is governed by reachability (Lemma~\ref{lem:reachable}), and thus exhibits  submodularity. Therefore, the combined propagation and adoption behavior in \model exhibits a complicated behavior that is neither submodular nor supermodular. In the next section, we will show that surprisingly, despite such complicated behavior, we can still design a greedy algorithm that achieves a  $(1-1/e-\epsilon)$-approximation to optimal expected social welfare.

\subsection{Approximation Algorithm}\label{sec:algo}
\eat{ 
\note[Laks]{suggested org of this section. \\ 
- main result \\ 
  + greedy algorithm \\ 
  + main theorem \\ 
- social welfare accounting technique \\ 
  + technique \\ 
  + proof \\
- preserving prefix property \\ 
  + primm algorithm 
  + proof of correctness and running time complexity. \\ 
- putting it all together 
  + summary of how pieces fit together} 
} 

\subsubsection{Greedy algorithm overview}\label{sec:overview} 

Given that the welfare function is neither submodular nor supermodular, designing an approximation algorithm for WelMax is challenging. Nevertheless, in this section we show that for any given $\epsilon>0$ and number $\ell\ge 1$, a $(1-\frac{1}{e} - \epsilon)$-approximation to the optimal social welfare can be achieved with probability at least $1 - \frac{1}{|V|^{\ell}}$, using a simple greedy algorithm. To the best of our knowledge, this is the first instance in the context of viral marketing where an {\sl efficient approximation algorithm is proposed for a non-submodular objective, at the same level as submodular objectives}. We first present our algorithm and then analyze its correctness and efficiency.

Our algorithm, called $\algo$ (for bundle greedy) and shown in Algorithm \ref{alg:greedy_allocation}, is based on a greedy allocation of seed nodes to items. Given a graph $G$, the universe of items $\allitems$, item  budget vector $\bvec$, $\epsilon$, and $\ell$, $\algo$ first selects (line~\ref{alg:PRIMM}) the top-$\bmax$ seed nodes $\greedSeeds := S_{\bmax}$ for the IC model (disregarding item utilities), where $\bmax = max\{b_i \mid b_i\in \bvec \}$. Then, (line~\ref{alg:assign}) for each item $i$ with budget $b_i$, it assigns the top-$b_i$ nodes from $\greedSeeds$ to $i$. We will show that this allocation achieves a $(1-\frac{1}{e} - \epsilon)$-approximation to the optimal expected social welfare. 
For this to work, the seed selection 
algorithm must ensure that the $\bmax$ seeds selected, $S_{\bmax}$,  satisfy a \emph{prefix-preserving} property (definition in \textsection\ref{sec:PRIMM}). \chgins{That is, intuitively, for every budget $b_i \in \bvec$, the top-$b_i$ seeds among $\greedSeeds$ must provide a $(1-\frac{1}{e} - \epsilon)$-approximation to the optimal expected spread under budget $b_i$. This property is ensured by invoking the $\PRIMM$ algorithm (Algorithm~\ref{alg:IMM_adopted}) in line~\ref{alg:PRIMM} of Algorithm~\ref{alg:greedy_allocation}.} 
The following is the main result for the $\algo$ algorithm.

\begin{theorem}\label{thm:main}
	Let $\greedAlloc$ be the greedy allocation  generated by $\algo$, 
	and $\optAlloc$	be the optimal allocation. 
	Given $\epsilon>0$ and $\ell>0$, with probability at least $1 - \frac{1}{|V|^{\ell}}$, 
	we have 
	\begin{equation} \label{eq:approxsw}
	\sw(\greedAlloc) \geq (1 - \frac{1}{e} - \epsilon ) \cdot \sw(\optAlloc).
	\end{equation}
The running time is 
$O((\bmax+\ell+ \log_n |\bvec|)(m+n)\log n / \epsilon^2)$.
\end{theorem} 

%
We note that our $\algo$ algorithm has the interesting property that it does not need the valuation functions,
	prices, and the distributions of noises as input, and thus works for all possible utility settings. 
It reflects the power of bundling --- as long as we know that all items are mutually complementary, then bundling them together as much as possible would always provide a good solution in terms of social welfare, no matter what the actual utilities. This is in stark contrast with the algorithmic solution
	in \cite{lu2015} for the complementary setting. Further, known algorithms for social welfare maximization in the combinatorial auction literature typically assume a value oracle (e.g., see~\cite{feige-vondrak-demand-2010, kapraov-etal-greedy-opt-soda-2013, korula-etal-online-swm-arxiv-2017}), which given a query as an itemset, returns the utility of the itemset. Works on IM for complementary items \cite{lu2015}, require the knowledge of adoption probabilities of every item given already adopted item subsets.
However, such an oracle can be quite expensive to realize in practice for non-additive utility functions, since there are exponentially many itemsets. 
In \textsection\ref{sec:blocks}, we show the approximation guarantee of our algorithm
	through the novel block accounting method, then
	in \textsection\ref{sec:PRIMM} we describe the prefix preserving influence maximization algorithm $\PRIMM$.
Algorithm~\ref{alg:IMM_adopted} is described and its correctness and running time complexity are established in \textsection\ref{sec:PRIMM}.

\begin{algorithm}[t!]
{ 
\DontPrintSemicolon
\caption{  $\algo(\allitems, \bvec, G, \epsilon, \ell)$} \label{alg:greedy_allocation}
$\greedAlloc \leftarrow \emptyset$; \\ 
 \eat{ 
\chgdel{Let $\greedSeeds := S^{Grd}_{\bmax}$ be the top-$\bmax$ seeds selected using a 
	prefix-preserving influence maximization algorithm under the IC model over graph $G$
	and budget vector $\bvec$, 
	i.e.,\\}
} 
	 $\greedSeeds \leftarrow \PRIMM(\bvec, G, \epsilon, \ell)$  \label{alg:PRIMM} \\
\For{$i \in \allitems$}  {
	Assign item $i$ to the first $b_i$ nodes of the ranked set $\greedSeeds$, i.e., \label{alg:assign} 
 $\greedSeeds_i \leftarrow \mbox{ top } b_i \mbox{ nodes from } \greedSeeds$  \;
	$\greedAlloc \leftarrow \greedAlloc \cup (\greedSeeds_i \times \{i\})$ \;}	
\textbf{return} $\greedAlloc$ as the final allocation\;  
}
\end{algorithm}


\subsubsection{Block accounting to analyze $\algo$} \label{sec:blocks} 

The analysis of the algorithm is highly non-trivial, because it needs to consider all possible seed allocations, propagation scenarios, with budgets possibly being non-uniform among items.
Our main idea is a ``block'' based accounting method: we break the set of items into a sequence of ``atomic'' blocks, such that each block has non-negative marginal utility given previous blocks, and it can be counted as an atomic unit in the diffusion process. 
Then we account for each block's contribution to the social welfare during a  propagation, and argue that for every block, the contribution of the block achieved by the greedy allocation is always at least $(1-1/e - \epsilon)$ times the contribution under any allocation. In \textsection 4.2.2.1 we first introduce the block generation process. Then using block based accounting, in \textsection 4.2.2.2 we establish the welfare produced by $\algo$, and later in \textsection 4.2.2.3, show an upper bound on the welfare produced by any arbitrary allocation.
The technical subtlety includes properly defining the blocks, showing why each block can be accounted for as an atomic unit
	separately, dealing with partial item propagation within blocks, etc.

In the rest of the analysis, we fix the noise world $W^N$, and prove that
$\sw_{\ow}(\greedAlloc) \geq (1 - \frac{1}{e} - \epsilon ) \cdot \sw_{\ow}(\optAlloc)$, where $\sw_{\ow}$ denotes the expected social welfare under the fixed noise world $\ow$.
We could then simply take another expectation over the distribution of $\ow$ to obtain  Inequality~\eqref{eq:approxsw}. 
Let $\utilow$ be the utility function under the noise possible world $\ow$.

Given $W^N$, let $\Iow \subseteq \allitems$ be the subset of items that gives the largest utility in $\ow$, with ties broken in favor of larger sets. 
By Lemma~\ref{lem:itemunion}, $\Iow$ is unique.
This implies that the marginal utility of any (non-empty) subset of $\allitems \setminus \Iow$ given $\Iow$ is strictly negative. 
Further recall that $\utilow$ is supermodular. 
Hence the marginal utility of any subset of $\allitems \setminus \Iow$ given any subset of $\Iow$ is strictly negative, which means no items in $\allitems \setminus \Iow$ can ever be adopted by any user under the noise world $\ow$.
Thus, once we fix $\ow$, we can safely remove all items in $\allitems \setminus \Iow$ from consideration.
In the rest of \textsection\ref{sec:blocks}, for simplicity we use $\Istar$ as a shorthand
	for $\Iow$.

\spara{4.2.2.1 Block generation process}
We divide items in $\Istar$ into a sequence of disjoint blocks such that each block has a non-negative marginal utility w.r.t. the union of all its preceding blocks.
We also need to carefully arrange items according to their budgets for later accounting analysis.
 We next discuss how the blocks are generated.

Let $\Istar = \{i_1, ..., i_{|\Istar|}\}$. We order the items in non-increasing order of their budgets, i.e., $b_1 \geq b_2 \geq \cdots \ge b_{|\Istar|}$.

{\color{black}
Figure \ref{alg:blockset-ucomic} shows the process of generating the blocks. Note that this {\sl block generation process is solely used for our accounting analysis and 
	is not part of our seed allocation algorithm. Hence it has no impact on the running time whatsoever
}  
Given $\Istar$ and $\ow$, we first generate a global sequence $\mathcal{I}$ of all non-empty subsets of $\Istar$, following a precedence order $\prec$ (Step {\bf 2}), explained next. 

For any two distinct subsets $S, S' \subseteq \Istar$, arrange items in each of $S, S'$ in decreasing order of item indices. Compare items in $S, S'$, starting from the highest indexed items of $S$ and $S'$. If they match then compare the second highest indexed items and so on until one of the following rules applies: \\ 
1. One of $S$ or $S'$ exhausts. If say $S$ exhausts first, then $S \prec S'$. \\
2. The current pair of items in $S$ and $S'$ do not match. Then $S \prec S'$, if the current item of $S$ has a lower index than the current item of $S'$.

We illustrate this step using the following example. 
}

\begin{figure}[t]
	
	\begin{framed}
		{

			\begin{description}[style=unboxed,leftmargin=8pt]
				\item[1.]
				Input for the process contains $\Istar$ and $\ow$.\\
				\item[2.]
				Generate the $2^{|\Istar|} -1$ non-empty subsets of $\Istar$ \\ 
				Sort the subsets following the precendence order $\prec$. Put the sorted subsets in sequence $\mathcal{I}$\\ 
				$\allblocks \leftarrow \emptyset$; $B \leftarrow$ the first entry in $\mathcal{I}$\\ 
				\item[3.]
				Repeat the following steps until $\mathcal{I}$ is 
empty \\
				\ \ (1) \emph{If} $\utilow(B \mid \bigcup \allblocks) \geq 0$ then, \\
				\ \ \ \ $\allblocks \leftarrow \allblocks \oplus B$ i.e., append $B$ at the end of sequence $\allblocks$ \\  
				\ \ \ \ remove all sets $B'$ from $\mathcal{I}$ with  $B' \cap B\ne \emptyset$ \\
				\ \ \ \ $B \leftarrow$ the first entry in $\mathcal{I}$ \\
				
				\ \ (2) \emph{Else}
				\ \ \ \ $B \leftarrow$ the next entry in $\mathcal{I}$ after $B$ \\	
				\item[4.]
				$\allblocks$ is the final sequence of blocks  
			\end{description}
			}
	\end{framed}
	\caption{The block generation process}
	\label{alg:blockset-ucomic}
\end{figure} 

\begin{example}[Generation of $\mathcal{I}$] \label{ex:sequence} 
{\em  
Suppose we have three items $\Istar$ $=$ $\{i_1, i_2, i_3\}$ with $b_1$ $\geq$ $b_2$ $\geq$ $b_3$, then we order the subsets in the following way: 
$\mathcal{I}$ $=$ $(\{i_1\}$, $\{i_2\}$, $\{i_1, i_2\}$, $\{i_3\}$, $\{i_1, i_3\}$, $\{i_2, i_3\}$, $\{i_1, i_2, i_3\})$. Between subsets $\{i_3\}$ and $\{i_1,i_3\}$, $\{i_3\}$ is ordered first 
	according to rule $1$, 
	whereas between $\{i_1,i_2\}$ and $\{i_3\}$, $\{i_1,i_2\}$ is ordered first according to rule $2$. \qed 
} 
\end{example} 

The sequence $\mathcal{I}$ has the following useful property:

\begin{property} \label{prop:subset}
	For any subsets $S$ and $T$ in the sequence $\mathcal{I}$, if (a) $T$ is a proper subset of $S$, or (b) the highest index among all items in $T$
		is strictly lower than the highest index among all items in $S$, then 
		$T$ appears before $S$ in $\mathcal{I}$.
\end{property}

From $\mathcal{I}$, blocks are selected following an iterative process, as shown in Step {\bf 3} of Figure~\ref{alg:blockset-ucomic}. 
We scan through this sequence, with the purpose of generating a sequence $\allblocks$ of disjoint blocks.
For each subset $B$ being scanned, if its marginal utility given all previously selected blocks is
	non-negative, i.e., $\utilow(B \mid \bigcup \allblocks) \geq 0$, where $\allblocks$ is the currently selected
	sequence of blocks, and $\bigcup \allblocks$ is the union of all items in these selected blocks, 
	then we append $B$ to the end of selected sequence $\allblocks$, i.e.,  $\allblocks = \allblocks \oplus B$, where $\oplus$ denotes ``append''. 
After selecting $B$, we remove all subsets in $\mathcal{I}$ that overlap with $B$, and restart the scan from the beginning of the remaining sequence.
If $\utilow(B \mid \bigcup \allblocks) < 0$, then we skip this set and go to the next one.

Example \ref{ex:bgen} illustrates the process.

\begin{example}[Block generation] \label{ex:bgen} 
{\em 
Continuing from Example~\ref{ex:sequence}, assume the following utility assignments for noise world $\ow$:
\begin{align*}
\utilow(i_1) = \utilow(i_2) = \utilow(i_3) = \utilow(i_1,i_2) = -1 \\
\utilow(i_1,i_3) = \utilow(i_2,i_3) = 1 ; \ \utilow(i_1,i_2,i_3) = 4
\end{align*}
Then as per the block generation process, $\{i_1,i_3\}$ will be chosen as the first block $B_1$, since it is the 
	first block in $\mathcal{I}$ with non-negative marginal utility w.r.t. $\emptyset$.
Once $B_1$ is chosen all itemsets containing $i_1$ or $i_3$ are deleted from $\mathcal{I}$, thus only $\{i_2\}$ remains in $\mathcal{I}$. Since $\utilow(\{i_2\} \mid \{i_1,i_3\}) = \utilow(i_1,i_2,i_3) - \utilow(i_1,i_3) = 4 -1 > 0$, $\{i_2\}$ is chosen as $B_2$ and the process terminates with 
$\allblocks = (\{i_1,i_3\}, \{i_2\})$. \qed }
\end{example}

By the fact that $\Istar$ is a local maximum, it is easy to see that the blocks generated form a partition of $\Istar$.
Let $\LL{\allblocks = } \{B_1, B_2, \ldots, B_t\}$ be the sequence of blocks generated, where $t$ is
	the number of blocks in the block partition.
We define the marginal gain of each block $B_i$ as
\begin{equation}
\Delta_i = \utilow(B_i \mid \cup_{j=1}^{i-1} B_j). \label{eq:Delta}
\end{equation}
We have the following properties regarding the marginal gains.

\begin{property}\label{prop:block_utility}
 $\forall i$ $\in$ $[\numblocks]$, $\Delta_i$ $\ge$ $0$ , and $\utilow(\Istar)$ $=$ $\sum_{i=1}^{\numblocks} \Delta_i$.
\end{property}

Let $A \subseteq \Istar$ be an arbitrary subset of items. 
We partition $A$ based on block partition $\allblocks$: Define $A_i = A \cap B_i, \forall i \in [\numblocks]$. 
If $A_i = B_i$, we call $A_i$ a full block, if $A_i = \emptyset$, then it is an empty block, otherwise, we call it a partial block. 
Define $\Delta^A_i = \utilow(A_i \mid A_1 \cup \ldots \cup A_{i-1})$. 
By Property \ref{prop:subset} and the fact that $B_i$ is the first block in $\mathcal{I}$ with non-negative marginal utility \LL{w.r.t. $\bigcup_{j=1}^{i-1} B_j$,} it follows that

\begin{property}\label{prop:block_utility2}
$\forall i$ $\in$ $[\numblocks]$, $\Delta^A_i$ $\le$ $\Delta_i$, and $\utilow(A)$ $=$ $\sum_{i=1}^t \Delta^A_i$.
\end{property}

Using this property, we devise our accounting where each $A_i$ contributes $\Delta^A_i$ 
	in its social welfare.

\spara{4.2.2.2 Social welfare under greedy allocation}
We are now ready to analyze the social welfare of our greedy allocation (Algorithm \ref{alg:greedy_allocation}) using block accounting. 
We first show that, before the propagation starts, each seed node would adopt exactly the prefix of full blocks allocated until
	the first non-full block, and then show that all these adopted
	full blocks will propagate together, so we can exactly account for the contribution of each block
	to the expected social welfare.
The following lemma gives the exact statement of the first part.

\begin{lemma} \label{lem:prefixfullblocks}
	Under the greedy allocation, 
 	suppose that at a seed node $v$, $A_i$ is the first non-full block assigned to $v$, 
 	then before the propagation starts, $v$ adopts exactly $B_1 \cup ... \cup B_{i-1}$.
 	
\end{lemma}

\begin{proof}
	This proof relies on the supermodularity of $\utilow$, the block generation process, the
		greedy allocation procedure, and Property~\ref{prop:block_utility2}.
	Let $M$ be the set of items adopted by $v$ before the propagation starts, and
	let $M_1 = M \cap (B_1 \cup ... \cup B_{i-1})$ and $M_2 = M \setminus M_1$.
	Since $A_i$ is a partial block, we know that $M_2 \ne B_i$.
	We first show that $M_2 = \emptyset$ and then 
		$M_1 = B_1 \cup ... \cup B_{i-1}$.
	
	Suppose, for a contradiction, that $M_2 \ne \emptyset$.
	We know that $\utilow(M_2 \mid M_1) \ge 0$, and by supermodularity
		$\utilow(M_2 \mid B_1 \cup ... \cup B_{i-1}) \ge 0$.
	If $M_2$ is ordered before $B_i$ in sequence $\mathcal{I}$, then $M_2$ should be selected
		instead of $B_i$, a contradiction.
	If $M_2$ is ordered after $B_i$ in $\mathcal{I}$, by the block generation process we can  
		conclude that all items in $B_i$ have budgets no less than the minimum budget for items
		in $M_2$, which by greedy allocation implies that all items in $B_i$ should be allocated
		to $v$, contradicting the fact $A_i$ is a partial block.
	\LL{Thus $M_2 = \emptyset$ and $M = M_1 \subseteq B_1 \cup ... \cup B_{i-1}$.}	
	
	Next, by Property~\ref{prop:block_utility2},\\
	 $$\LL{\utilow(M) = \sum_{j=1}^{i-1} \Delta^M_i
		\leq} \sum_{j=1}^{i-1} \Delta_i = \utilow(B_1 \cup ... \cup B_{i-1})$$ Thus
		$v$ should adopt $B_1 \cup ... \cup B_{i-1}$ instead of $M$.
\end{proof}

\spara{Effective budget of blocks} 
For a block $B_i$, we define its {\em effective budget} $\bbude_i = \min_{j \in B_1 \cup \cdots \cup B_i} b_j$.  
	In $\algo$ (Algorithm \ref{alg:greedy_allocation}), the first $\bbude_i$ seed nodes of $\greedSeeds$ are assigned all the full blocks $\{B_1 \cup ... \cup B_i\}$. By Lemma \ref{lem:prefixfullblocks}, only those nodes actually adopt the block $B_i$ 
	before the propagation starts. 
	Such seed nodes are called \emph{effective seed nodes} of block $B_i$ and denoted as $\greedESeeds_{B_i}$. Thus in summary, under the greedy allocation, before the propagation starts, 
	all seed nodes in $\greedESeeds_{B_i}$ adopt $B_i$ together with $B_1, \ldots, B_{i-1}$, and
	none of the seed nodes outside $\greedESeeds_{B_i}$ adopts 
	any items in $B_i, B_{i+1}, \ldots, B_t$. 
	
As established, the nodes in $\greedESeeds_{B_i}$ always adopt $B_i$ together with $B_1, \ldots, B_{i-1}$ and without considering the effect of propagation, no other seed nodes outside the set $\greedESeeds_{B_i}$ adopts $B_i$ or any other blocks $B_{i+1}, \ldots, B_t$. $B_i$ is not adopted because at least one of the previous $B_1, ..., B_{i-1}$ blocks is not allocated to those nodes. Also since $B_i$ is not adopted, none of the subsequent blocks can be adopted. We illustrate this using an example next.

\begin{example}[Block budgets] \label{ex:bbud} 
{\em 
Revisit the blocks shown in Example \ref{ex:bgen}. Let us assume that $b_1 > b_2 > b_3$. Recall that $B_1 = \{i_1,i_3\}$ and $B_2=\{i_2\}$. 
Let $S_2,S_3$ be the top $b_2,b_3$ nodes in the greedy allocation respectively, and
		$S_3 \subset S_2$.
Then under the greedy allocation, $B_2$ as a full block will be allocated to nodes in $S_2$. 
The effective budget of $B_2$ is $\bbude_2 = \min_{j \in B_1 \cup B_2} b_j = b_3$.
The effective seed set of $B_2$ is $\greedESeeds_{B_2} = S_3$, since nodes in $S_3$ are allocated
	both $B_1$ and $B_2$ and will adopt both $B_1$ and $B_2$ according to Lemma~\ref{lem:prefixfullblocks} (can also be verified by checking the utility settings
	given in Example \ref{ex:bgen} manually).
For nodes in \LL{$S_2 \setminus S_3$,} even though they are allocated the full block $B_2$, 
	they are only allocated a partial block $A_1 = \{i_1\}$, and thus by Lemma~\ref{lem:prefixfullblocks} they will not adopt $B_2$ or $A_1$.
\qed
}
\end{example}

We are now ready to show the social welfare of the allocation made by $\algo$.

\begin{restatable}{lemma}{lemgreedyasw}
	\label{lem:greedyasw} 
	Let $\greedAlloc$ be the greedy allocation obtained using Algorithm \ref{alg:greedy_allocation}. Then the expected social welfare of $\greedAlloc$ in $\ow$ is $\sww(\greedAlloc) = \sum_{i\in [t]} \sigma(\greedESeeds_{B_i}) \cdot \Delta_i$, where $\greedESeeds_{B_i}$ are the effective seed nodes of block $B_i$ under allocation $\greedAlloc$, $\sigma(\cdot)$ is the expected
	spread function under the IC model, and $\Delta_i$ is as defined in Eq.~\eqref{eq:Delta}.
\end{restatable}

\begin{proof}
To account for the effect of propagation, we use the Reachability Lemma (Lemma \ref{lem:reachable}). By that lemma, nodes reachable from $\greedESeeds_{B_i}$ adopt all the blocks $B_1, ..., B_i$. For a full block $B_i$ only the effective seeds of $B_i$ and nodes reachable from them adopt $B_i$. Thus the expected number of nodes that are reached by block $B_i$ and consequently adopt $B_i$, is $\sigma(\greedESeeds_{B_i})$. \LL{From Property \ref{prop:block_utility}, adoption of every such $B_i$ contributes $\Delta_i$ to the overall social welfare. Moreover, the only item adoptions are disjoint union of full blocks. Hence $\sww(\greedAlloc) = \sum_{i\in [t]} \sigma(\greedESeeds_{B_i}) \cdot \Delta_i$.}
\end{proof}

\spara{4.2.2.3 Social welfare under an arbitrary allocation}

Unlike greedy, in an arbitrary allocation, for the effective seed nodes, we cannot conclude that a block $B_i$ is offered with all previous full blocks $B_1, \ldots, B_{i-1}$. Thus our accounting method needs to be adjusted.
Our idea is to define the key concept of an {\em anchor item} $a_i$ for every block $B_i$, which appears in $B_1\cup \cdots \cup B_i$.
We want to show that only when $B_i$ is co-adopted with $a_i$ by any node, $B_i$ could contribute positive marginal social welfare (Lemma~\ref{lem:anchor}),
	and in this case its marginal contribution is upper bounded by $\Delta_i$ (Property~\ref{prop:block_utility2}).
Hence we only need to track the diffusion of the anchor item $a_i$ to account for the marginal contribution of $B_i$.
Finally by showing that the budget
	of $a_i$ is exactly the effective budget $\bbude_i$ $=$ $|\greedESeeds_{B_i}|$ of $B_i$, we conclude that 
	$\sigma(\allseeds_{a_i})$ $\le$ $(1-1/e-\epsilon)$ $\sigma(\greedESeeds_{B_i})$ by the prefix preserving property explained in \textsection\ref{sec:overview}.

We define the budget of a block to be the minimum budget of any item in the block. Then the 
\emph{anchor block} $B^a_i$, of a block $B_i$ is the block from $B_1, \ldots, B_{i}$ that has the minimum budget. 
In case of a tie, the block having highest index is chosen as the anchor block. 
Notice that anchor item $a_i$ is the highest indexed and consequently minimum budgeted item in its corresponding anchor block $B^a_i$. 
Notice that, by definition, if block $B_j$ is the anchor block of block $B_i$ with $j< i$, then
	block $B_j$ is also the anchor block for all blocks $B_j, B_{j+1}, \ldots, B_i$.
Moreover, the effective budget $\bbude_i$ of a block $B_i$, is the budget of its anchor item $a_i$, i.e., the minimum budget of all items in $B_1 \cup \cdots \cup B_i$.
We illustrate the concept of anchor block and item using the example below.

\begin{example}[Anchor block and item] \label{ex:anc} 
{\em 
Anchor block of block $B_2$ in Example \ref{ex:bbud}, is $B^a_2 = B_1$. Its corresponding anchor item $a_2$ is the highest indexed item of block $B^a_2$, i.e., $i_3$. Block $B_1$'s anchor block is the block itself and consequently its anchor item $a_1$ is again $i_3$. \qed } 
\end{example}

\begin{restatable}{lemma}{lemanchor}
	\label{lem:anchor}
Let $a_i$ be the anchor item of $B_i$, and suppose \LL{$a_i$} appears in $B_j$, $j\le i$.
During the diffusion process from an arbitrary seed allocation $\allalloc$, 
let $A$ be the set of items in $B_{j} \cup \ldots \cup B_i$
that have been adopted by $v$ by time $t$.
If $a_i \notin A$ and $A \ne \emptyset$, then 
$\utilow(A \mid B_1 \ldots, B_{j-1}) < 0$.  
\end{restatable}

\begin{proof}
	Suppose that $\utilow(A \mid B_1 \cup \cdots \cup B_{j-1}) \ge 0$.
	By the definition of the anchor item, we know that all items in $A \setminus B_j $ have strictly larger budget than the budget
	of $a_i$, otherwise one of items in $A\setminus B_j $ should be the anchor item for $B_i$. 
	This means all items in $A \setminus B_j $ have index strictly lower than $a_i$.
	Notice $a_i \notin A$, and thus all items in $A \cap B_j$ also have index strictly lower than $a_i$.
	Then by Property \ref{prop:subset}, $A$ should appear before $B_j$ in sequence $\mathcal{I}$.
	Since $\utilow(A \mid B_1 \ldots, B_{j-1}) \ge 0$, the block generation process should select $A$ as the
	$j$-th block instead of the current $B_j$, a contradiction. 
\end{proof}

Using the above result, we establish the following lemma, which upper bounds the welfare produced by an arbitrary allocation.

\begin{restatable}{lemma}{lemallasw}
	\label{lem:allasw}
	For any arbitrary seed allocation $\allalloc$, the expected social welfare in $\ow$ is $\sww(\allalloc) \leq \sum_{i\in [t]} \sigma(\allseeds_{a_i}) \cdot \Delta_i$, where $\allseeds_{a_i}$ is the seed set assigned to the anchor
	item $a_i$ of block $B_i$, and $\Delta_i$ is as defined in Eq.~\eqref{eq:Delta}. 
\end{restatable}

\begin{proof}
	For an edge possible world $W^E$, 
	suppose that after the diffusion process under $W^E$, every node $v$ adopts item set $A_v$.
	Let $A_{v,i} = A_v \cap B_i$ for all $i \in [t]$, and
	$\Delta^{A_v}_i = \utilow(A_{v,i} \mid A_{v,1} \cup \ldots \cup A_{v,i-1})$.
	Thus, we have

	\begin{align}
	\sww(\allalloc) & = \E_{W^E} \left[\sum_{v\in V} \utilow(A_v) \right] 
	= \E_{W^E} \left[\sum_{v\in V} \sum_{i\in [t]} \Delta^{A_v}_i \right]  \nonumber \\
	&	= \sum_{i\in [t]} \E_{W^E} \left[ \sum_{v\in V} \Delta^{A_v}_i \right], \label{eq:welfareEquiv2}
	\end{align}
	where the expectation is taken over the randomness of the edge possible worlds, and thus we use
	subscript $W^E$ under the expectation sign to make it explicit.
	By switching the summation signs and the expectation sign in the last equality above, 
	we show that the expected social welfare can be
	accounted as the summation among all blocks $B_i$  of the expected marginal gain of block $B_i$ on all nodes.
	We next bound $\E_{W^E} \left[\sum_{v\in V} \Delta^{A_v}_i \right]$ for each block $B_i$.
	
	Under the edge possible world $W^E$, for each $v\in V$, there are three possible cases for $A_{v,i}$.
	In the first case, $A_{v,i} = \emptyset$. In this case, $\Delta^{A_v}_i = 0$, so we do not need to count 
	the marginal gain $\Delta^{A_v}_i$.
	In the second case, $A_{v,i}$ is not empty but
	it does not co-occur with block $B_i$'s anchor $a_i$, that is $a_i \not\in A_v$, and $A_{v,i} \ne \emptyset$.
	In this case, Let $A' = A \cap (B_j \cup \ldots \cup B_i)$, where $B_j$ is the anchor block of $B_i$.
	Then $A'$ is not empty and we know $\utilow(A' \mid B_1 \cup \ldots \cup B_{j-1}) < 0$.
	Since we have $\utilow(A' \mid B_1 \cup \ldots \cup B_{j-1})  
	= \sum_{j'=j}^i \Delta^{A_v}_{j'}$. Thus the cumulative marginal gain of $\Delta^{A_v}_{j'}$ with $j\le j' \le i$ is negative, so 
	we can relax them to $0$, effectively not counting the marginal gain of $\Delta^{A_v}_i$ either.
	
	Finally, $A_{v,i}$ is non-empty and co-occur with its anchor $a_i$, 
	i.e. $a_i \in A$ and $A_{v,i} \ne \emptyset$.
	Since $A_v$ is a partial block, $\Delta^{A_v}_i \le \Delta_i$, we relax
	$\Delta^{A_v}_i$ to $\Delta_i$.
	This relaxation occurs only on nodes that adopt $a_i$.
	A node $v$ could adopt $a_i$ only when there is a path in $W^E$ 
	from a seed node that adopts $a_i$ to node $v$. 
	As defined in the lemma, $S_{a_i}$ is the set of seed nodes of $a_i$.
	Let $\Gamma(S_{a_i}, W^E)$ be the set of nodes that are reachable from $S_{a_i}$ in $W^E$.
	Then, there are at most $| \Gamma(S_{a_i}, W^E)|$ nodes at which we relax $\Delta^{A_v}_i$ to $\Delta_i$ 
	for block $B_i$. Hence,
	
	\begin{align} \label{eq:marginalRelax2}
	\sum_{v\in V} \Delta^{A_v}_i \le | \Gamma(S_{a_i}, W^E)| \cdot \Delta_i.
	\end{align}
	
	Furthermore, notice that $\E_{W^E}[| \Gamma(S_{a_i}, W^E)| ]$ $=$ $\sigma(S_{a_i})$, by the live-edge representation of
	the IC model.
	Therefore, together with Eq.~\eqref{eq:welfareEquiv2} and~\eqref{eq:marginalRelax2}, we have
	\begin{align*}
	\sww(\allalloc) & \le \sum_{i\in [t]} \E_{W^E} \left[  | \Gamma(S_{a_i}, W^E)| \cdot \Delta_i \right] \\ 
	& =  \sum_{i\in [t]} \sigma(S_{a_i}) \cdot \Delta_i.
	\end{align*}
	This concludes the proof of the lemma.
\end{proof}

Notice in Lemma~\ref{lem:allasw}, $|\allseeds_{a_i}| \leq \bbude_i$, whereas in Lemma~\ref{lem:greedyasw} $|\greedESeeds_{B_i}| = \bbude_i$. Hence
the combination of Lemma~\ref{lem:greedyasw} and Lemma~\ref{lem:allasw}, together with the fact that $\greedESeeds_{B_i}$ is a $(1-1/e-\epsilon)$-approximation of
	the optimal solution with $\bbude_i$ seeds (by the prefix-preserving property),
	leads to the approximation guarantee of $\algo$ (Eq. (\ref{eq:approxsw}) of Theorem~\ref{thm:main}), which we prove next.

\begin{restatable}{theorem}{thmnonuniformbudget} {\sc (Correctness of $\algo$)\ \ }\label{thm:nonuniformbudget}
	Let $\greedAlloc$ be the greedy allocation and $\allalloc$ be any arbitrary allocation.
	Given $\epsilon>0$ and $\ell>0$, the expected social welfare $\sw(\greedAlloc) \geq (1 - \frac{1}{e} - \epsilon ) \cdot \sw(\allalloc)$ with at least $1 - \frac{1}{|V|^{\ell}}$ probability.
\end{restatable}
\begin{proof}
	From Lemma \ref{lem:greedyasw} , we have for a possible world $\ow = (W^E, W^N)$,
	$\sww(\greedAlloc) =\sum_{i\in [\numblocks]} \sigma(\greedESeeds_{B_i}) \cdot \Delta_i $,
	where the size of $\greedESeeds_{B_i}$ is the effective budget of $B_i$.
	
	For an arbitrary allocation $\allalloc$, 
	since $a_i$ is the anchor item of $B_i$, by its definition 
	we know that $|\allseeds_{a_i}| = |\greedESeeds_{B_i}|$. 
	By the correctness of the prefix-preserve influence maximization algorithm we use in line~\ref{alg:PRIMM} 
	(Definition~\ref{def:PRIMM}, to be instantiated in \textsection\ref{sec:PRIMM}), we have that 
	with probability at least $1 - \frac{1}{|V|^{\ell}}$, 
	$\sigma(\greedESeeds_{B_i}) \geq (1 - \frac{1}{e} - \epsilon) \sigma(\allseeds_{a_i})$, for all blocks $B_i$'s
	and their corresponding anchors $a_i$'s.
	
	Let the distribution of world $\ow$ be $\dow$. Then, together with Lemma~\ref{lem:allasw},
	we have that with probability at least $1 - \frac{1}{|V|^{\ell}}$,
	\begin{align*}
	\sw(\greedAlloc) &= \mathbb{E}_{\ow \sim \dow} [\sww(\greedAlloc)] \\
	&= \mathbb{E}_{\ow \sim \dow} \left[\sum_{i\in [\numblocks]} \sigma(\greedESeeds_{B_i}) \cdot \Delta_i \right]\\
	&\geq \mathbb{E}_{\ow \sim \dow} \left[\sum_{i\in [\numblocks]}(1 - \frac{1}{e} - \epsilon) \sigma(\allseeds_{a_i}) \cdot \Delta_i \right]\\
	& \geq (1 - \frac{1}{e} - \epsilon) \mathbb{E}_{\ow \sim \dow} [ \sww(\allalloc)]  \\
	& = (1 - \frac{1}{e} - \epsilon) \sw(\allalloc).
	\end{align*}
	Therefore, the theorem holds.
\end{proof}	
	
In the following section, we explain the component $\PRIMM$ that provides the prefix preserving property.

\subsubsection{Item-wise prefix preserving IMM} \label{sec:PRIMM}

We first formally define the prefix-preserving property.

\begin{definition}{\sc (Prefix-Preserving Property).\ } \label{def:PRIMM}
	\fixedspaceword[.75]{Given $G$ $=$ $(V,E,p)$ and budget vector $\bvec$, an influence maximization algorithm $\mathbb{A}$
	is prefix-preserving w.r.t. $\bvec$, 
	if for any $\epsilon$ $>$ $0$ and $\ell$ $>$ $0$, 
	$\mathbb{A}$  returns an ordered set $\greedSeeds_{\bmax}$ of size $\bmax$, 
	such that  with probability  at least $1$ $-$ $\frac{1}{|V|^{\ell}}$, 
	for every $b_i$ $\in$ $\bvec$, the top-$b_i$ nodes of $\greedSeeds_{\bmax}$, denoted $\greedSeeds_{b_i}$, satisfies $\sigma(\greedSeeds_{b_i})$ $\geq$ $(1-\frac{1}{e} - \epsilon)$ $\OPT_{b_i}$, 
	where 
	$\OPT_{b_i}$ is the optimal expected spread of $b_i$ nodes.}
\end{definition}

%

Unfortunately, state-of-the-art IM algorithms such as IMM \cite{tang15}, SSA \cite{Nguyen2016}, and OPIM \cite{xiaokui-opim-sigmod-2018} are not prefix-preserving out-of-the-box. 
In this section, we present a non-trivial extension of IMM \cite{tang15}, called $\PRIMM$ (PRefix preserving IM Algorithm) (Algorithm \ref{alg:IMM_adopted}), to make it prefix-preserving.
The classical models of influence propagation assume a single item and 
IMM is one of the state of the art algorithms for influence maximization. For a single item, as well as for multiple items with uniform budgets, the prefix property is trivial. In the presence of multiple items with non-uniform budgets, an algorithm that returns a seed set of high quality with only a probabilistic guarantee need \emph{not} satisfy the prefix preserving property (Definition \ref{def:PRIMM}). We present $\PRIMM$ (\em{PR}efix \em{IMM}), shown in Algorithm \ref{alg:IMM_adopted}, which is a  prefix-preserving extension of IMM for multiple items. Notice that $\nodeselect(\mathcal{R},k)$ is the standard greedy algorithm for finding a seed set of size $k$ by solving max $k$-cover on the set of RR sets  $\mathcal{R}$. For more details, the reader is referred to \cite{tang15}. 
The $\nodeselect$ algorithm used in $\PRIMM$ is same as Alg $1$ of IMM, which we donot repeat for brevity.  

\begin{algorithm}[t!]
{ 
\DontPrintSemicolon
\caption{$\PRIMM$ $(\bvec, G, \epsilon, \ell)$} \label{alg:IMM_adopted}
Initialize $\mathcal{R} = \emptyset$, $s = 1$, $n = |V|$, $i = 1$, $\epsilon^\prime = \sqrt{2} \cdot \epsilon$, $\budgetSwitch = \false$; \\
$\ell = \ell  + \log 2 / \log n $, $\ell^\prime = \log_n(n^{\ell} \cdot |\bvec|)$; \\ \label{lin:ellb}
\While{$i \leq \log_2(n) - 1$ and $s \leq |\bvec|$} {
	$k = b_s$, $LB =1$; \\
	$x = \frac{n}{2^i}$; $\theta_i = \lambda^\prime_k/x$, where $\lambda^\prime_k$ is defined in Eq. \eqref{eq:lambdap}; \\
	\While{$|\mathcal{R}| \leq \theta_i$} {                                  \label{lin:genr1}
		Generate an RR set for a randomly selected node $v$ of $G$ and insert in $\mathcal{R}$; \\ \label{lin:genr2}
	}
	\If{$\budgetSwitch$}{
		$S_k = $ the first $k$ nodes in the ordered set $S_{b_{s-1}}$ returned from
		the previous call to $\nodeselect$
	}
	\Else {
		$S_k =\nodeselect(\mathcal{R}, k)$
	}
	\If{$n \cdot F_{\mathcal{R}}(S_k) \geq (1 + \epsilon^\prime)\cdot x$} { \label{lin:cov}
		$LB = n \cdot F_{\mathcal{R}(S_k)} / (1+\epsilon^\prime) $; \\
		$\theta_k = \lambda^\ast_k / LB$, where $\lambda^\ast_k$ is defined in Eq. \eqref{eq:lambdaa}; \\ \label{lin:suc1}
		\While{$|\mathcal{R}|<\theta_k$} {
			Generate an RR set for a randomly selected node $v$ of $G$ and insert in $\mathcal{R}$; \\    \label{lin:suc2}
		}
		$s = s + 1 $; $\budgetSwitch = \true$ \\	
	}
	\Else{ $i= i+1$; $\budgetSwitch = \false$ \\} \label{lin:fail}
	
}
\If{$s \leq |\bvec|$}{
	$\theta_k = \lambda^\ast_{b_s}/LB$; \\ \label{lin:final}	
}	
$\mathcal{R} = \emptyset$; \\
\While{$|\mathcal{R}|<\theta_k$} {   \label{lin:weinote}
		Generate an RR set for a randomly selected node $v$ of $G$ and insert in $\mathcal{R}$; \\
}
$S_{\bmax} = \nodeselect(\mathcal{R}, \bmax)$; \\ \label{lin:nodes}
\textbf{return} $S_{\bmax}$ as the final seed set;
}
\end{algorithm}

	State-of-the-art IM algorithms including IMM use reverse influence sampling (RIS) approach~\cite{borgs14} governed by reverse-reachable (RR) sets.
	An RR set is a random set of nodes sampled from the graph by (a) first selecting a node $v$ uniformly at random from the graph, and
	(b) then simulating the reverse propagation of the model (e.g., IC model) and adding all visited nodes into the RR set.
	The main property of a random RR set $R$ is that: influence spread $\sigma(S) = n\cdot \mathbb{E}[\mathbb{I}\{S\cap R \ne \emptyset\}]$ for any seed set $S$, where
	$\mathbb{I}$ is the indicator function.
	After finding large enough number of RR sets, the original influence maximization problem is turned into a $k$-max coverage
	problem -- finding the set of $k$ nodes that covers the most number of RR sets, where a set $S$ covers an RR set $R$ if $S \cap R \ne \emptyset$.
	All RIS algorithms use the same well-known coverage procedure, denoted as $\nodeselect(\mathcal{R},k)$ in~\cite{tang15}, and thus we omit its
	description here.
	These algorithms mainly differ in estimating the number of RR sets needed for the approximation guarantee.
	The number of RR sets generated by these algorithms is in general not monotone with the budget $k$, making them not prefix preserving.
	Our $\PRIMM$ algorithm carefully addresses this issue, even with nonuniform item budgets, while keeping the efficiency of the algorithm.

$\PRIMM$ ingests four inputs, namely the budget vector $\bvec$, graph $G$, $\epsilon$ and $\ell$, with $\bvec$ sorted in non-increasing order as stated in Definition~\ref{def:PRIMM}.
Given $\ell$, for a budget $k$, IMM generates a set of RR sets $\mathcal{R}$, such that $|\mathcal{R}|\geq \lambda^\ast_k / \OPT_k$ with probability at least $1 - 1/n^{\ell}$. $\PRIMM$ derives a number $\ell^\prime > \ell$ as a function of $\ell$ (Algorithm \ref{alg:IMM_adopted}, line \ref{lin:ellb}), the details of which we provide in Lemma \ref{lem:primmcorrect}. Before that, we briefly describe $\PRIMM$. Extending the bounding technique of \cite{tang15}, for each budget $k$, we set 
\begin{equation}\label{eq:lambdap}
\lambda^\prime_k = \frac{(2+\frac{2}{3}\epsilon^\prime) \cdot (\log {n \choose k} + \ell^\prime \cdot \log \ n+\log\log_{2} \ n )\cdot n}{\epsilon^{\prime 2}},
\end{equation} 
\vspace*{-1mm}
\begin{gather}\label{eq:lambdaa}
\lambda^\ast_k = 2n \cdot ((1 - 1/e)\cdot \alpha + \beta_k)^2 \cdot \epsilon^{-2}, 
\end{gather} 
where, $\alpha = \sqrt{\ell' \log n + \log 2}$ is a constant independent of $k$, and 
$\beta_k = \sqrt{(1 - 1/e) \cdot (\log \tbinom{n}{k}+\ell^\prime \log \ n+ \log 2)}$.
Note that we use $\log$ without a base to represent the natural logarithm.

The basic idea of $\PRIMM$ is to generate enough RR sets such that for any budget $k \in \bvec$, $|\mathcal{R}| \geq \lambda^\ast_k / \OPT_k$, with probability at least $1 - 1/n^{\ell^\prime}$. Since $\OPT_k$ is unknown, we rely on a good lower bound of $\OPT_k$, i.e., $LB_k$, as proposed in IMM \cite{tang15}. Specifically $\PRIMM$ starts from the highest budget, i.e., $b_1$. For a given budget $k \in \bvec$ and $i$ it samples enough RR sets into $\mathcal{R}$ first (lines \ref{lin:genr1}-\ref{lin:genr2}) and then checks the coverage condition on the sampled set of RR sets (line \ref{lin:cov}). Note if $\mathcal{R}$ already had enough number of RR sets (generated at a previous budget), then it skips RR set generation and moves directly to coverage check. If the coverage condition succeeds, then a good $LB$ for the budget $k$ is determined. 
It uses the $LB$ to find the required number of RR sets (lines \ref{lin:suc1}-\ref{lin:suc2}) for $k$ and moves to the next budget.
It then reuses the prefix of the ordered seed set found for budget $k$ 
	as the seed set found for the new budget, avoiding a redundant call to the
	$\nodeselect$ procedure.
This is fine because $\nodeselect$ is a deterministic greedy procedure in finding
	seed nodes, and the last call to $\nodeselect$ before the budget switch, is using
	the same RR set collection $\mathcal{R}$ with a larger budget, and thus it already found
	all the seed nodes for the new budget.
If the coverage condition fails,
	 it increments $i$ to sample more RR sets for the current budget $k$ (line \ref{lin:fail}). 

If for any budget, all possible $i$ values are tested, $\PRIMM$ breaks the for-loop and generates RR sets (for that budget) using $LB=1$ (line \ref{lin:final}), which is the lowest possible value of $LB$. Further, since budgets are sorted in non-increasing order and $\lambda^\ast_k$ is monotone in $k$ (Eq. \eqref{eq:lambdaa}), there cannot be any remaining budget $k^\prime$, where $k^\prime \le k$, for which $\lambda^\ast_{k^\prime} /LB$ (line  \ref{lin:final}) is higher. Hence the RR set generation process terminates.

Lastly, after determining $|\mathcal{R}|$, those many RR sets are generated from scratch (line \ref{lin:weinote}) on which the final $\nodeselect$ is invoked. This addresses a recently found issue of the original IMM algorithm \cite{chen2018issue}.
	$\PRIMM$ then returns the top-$\bmax$ seeds obtained from $\nodeselect$ (line \ref{lin:nodes}).

The correctness and the running time of the $\PRIMM$ algorithm mainly follow 
	the  proof of the IMM algorithm~\cite{tang15,chen2018issue}. We first show the correctness and towards that we prove that the following lemma holds.

\begin{lemma}\label{lem:primm}
	Let $\mathcal{R}$ be the final set of RR sets generated by $\PRIMM$ at the end and let $k\in \bvec$ be any budget. Then $|\mathcal{R}|  \geq  \lambda^\ast_k / \OPT_k$ holds with probability at least $1 - 1/n^{\ell^\prime}$.
\end{lemma}

\begin{proof}
	Given $x \in [1,n]$, $\epsilon^\prime$ and $\delta_3 \in (0,1)$ and a budget $k$. Let $S_k$ be the seed set of size $k$ obtained by invoking $\nodeselect(\mathcal{R},k)$, where,
	\begin{equation}\label{eq:rrsize}
	|\mathcal{R}| \geq \frac{(2+\frac{2}{3}\epsilon^\prime)\cdot (\log {n \choose k} + \log (1/\delta_3))}{\epsilon^\prime} \cdot \frac{n}{x}.
	\end{equation} 
	Then, from Lemma $6$ of \cite{tang15}, if $\OPT_k < x$, then $n \cdot F_{\mathcal{R}}(S_k) < (1+\epsilon^\prime) \cdot x$ with probability at least $(1-\delta_3)$. Now let $j = \ceil{\log_2 \frac{n}{\OPT_k}}$. By union bound, we can infer that $\PRIMM$ has probability at most $(j-1) / (n^{\ell^\prime} \cdot \log_2 n)$ to satisfy the coverage condition of line \ref{lin:cov} for the budget $k$. Then by Lemma $7$ of \cite{tang15} and the union bound, $\PRIMM$ will satisfy $LB_k \leq \OPT_k$ with probability at least $1 - n^{\ell^\prime}$. We know that for any $k \in \bvec$, $|\mathcal{R}| \geq \lambda^\ast_k / LB_k$, hence the lemma follows.
\end{proof}

We are now ready to prove the correctness of $\PRIMM$.

\begin{lemma} \label{lem:primmcorrect}
	$\PRIMM$ returns a prefix preserving $(1 - 1/e - \epsilon)$-approximate solution $S_{\bmax}$ to the optimal expected spread, with probability at least $1 - 1/n^\ell$. 
\end{lemma}

\begin{proof}
	We know from Lemma \ref{lem:primm} that the RR set sampling for any budget can result in the coverage condition (Algorithm~\ref{alg:IMM_adopted}, line~\ref{lin:cov}) failing with probability at most $1/n^{\ell^\prime}$. By applying union bound over all the budgets, we have that the failure probability of the coverage condition in $\PRIMM$ is at most $\sum_{k \in \bvec} 1/n^{\ell^\prime} = |\bvec| \cdot 1/n^{\ell^\prime}$. By setting $\ell^\prime = \log_n(n^{\ell} \cdot |\bvec|)$, we bound this failure probability to at most $1/n^\ell$. {\color{black} Thus $\ell^\prime$ is used for computing $\alpha$ and $\beta_k$ in Eq. (\ref{eq:lambdaa}). Further once $\theta_k$ is determined, we generate those many RR set from scratch. This follows the fix proposed in \cite{chen2018issue} for a bug in Theorem $1$ of \cite{tang15}. Without the fix, the top $S_{\bmax}$ nodes returned by the last call to $\nodeselect$ (line \ref{lin:nodes}), cannot be shown to have a $(1-1/e-\epsilon)$-approximate solution with probability at least $1-1/n^{\ell}$. For every budget $b_i\in\bvec$, we can then choose the prefix of top-$b_i$ nodes of $S_{\bmax}$ and use that as a solution $S_{b_i}$ for that budget, with the guarantee that with probability at least $1-1/n^{\ell}$ each $S_{b_i}$ is a $(1-1/e-\epsilon)$-approximate solution to $\OPT_{b_i}$.}
 
	By union bound, $\PRIMM$ returns a $(1 - 1/e - \epsilon)$-approximate prefix preserving solution with probability at least $1 - 2/n^{\ell}$.
	
	Finally by increasing $\ell$ to  $\ell + \log 2/ \log n$ in line \ref{lin:ellb}, we raise $\PRIMM$'s probability of success to $1 - 1/n^\ell$. 
\end{proof}

\noindent
{\bf Running time\ \ }

The running time of $\PRIMM$ essentially involves two parts: the time needed to generate the set of RR sets $\mathcal{R}$ and the total time of all $\nodeselect$ invocations. From Lemma $9$ of \cite{tang15}, we have for any budget $k$, the set of RR sets generated for that budget $\mathcal{R}_k$ satisfies, 

\begin{align*}
\mathbb{E}[|\mathcal{R}_k|] &\leq \frac{3 max\{ \lambda^\ast_k, \lambda^\prime_k\} \cdot (1+\epsilon^\prime)^2}{(1 - 1/e)\cdot \OPT_k}.
\end{align*}

Since $\lambda^\prime_k$ and $\lambda^\ast_k$ are both monotone in $k$ (Eq. (\eqref{eq:lambdap}) and (\eqref{eq:lambdaa})), we know their maximums are achieved for $k = \bmax$. 
 
Further let $\OPT_{min} := \OPT_{b_{|{\bf I}|}}$ be the minimum expected spread, i.e., minimum value of $\OPT$, across all budgets, then for any $\mathcal{R}_k$,
\begin{align*}
\mathbb{E}[|\mathcal{R}_k|] &\leq \frac{3 max\{ \lambda^\ast_{\bmax}, \lambda^\prime_{\bmax}\} \cdot (1+\epsilon^\prime)^2}{(1 - 1/e)\cdot \OPT_{min}} \\
& = O((\bmax+\ell^\prime)n\log \ n \cdot \epsilon^{-2} / \OPT_{min}).
\end{align*}

Further since $\PRIMM$ reuses the RR sets instead of generating them from scratch for every budget, for the RR set $\mathcal{R}$ generated by $\PRIMM$,

\begin{align}\label{eq:rsize}
\mathbb{E}[|\mathcal{R}|] &= max_{k \in \bvec}\{ \mathbb{E}[|\mathcal{R}_k|] \} \nonumber \\ 
& = O((\bmax+\ell^\prime)n\log \ n \cdot \epsilon^{-2} / \OPT_{min}).
\end{align}

For an RR set $R \in \mathcal{R}$, let $w(R)$ denote the number of edges in $G$ pointing to nodes in $R$. If $EPT$ is the expected value of $w(R)$, then we know, $n \cdot EPT \leq m \cdot \OPT_{min}$ \cite{tang15}. Hence using Eq.~\eqref{eq:rsize}, the expected total time to generate $\mathcal{R}$ is determined by,
\begin{align}\label{eq:rtime}
\mathbb{E}[\sum_{R \in \mathcal{R}}w(R)] &= \mathbb{E}[|\mathcal{R}|] \cdot EPT \nonumber \\
& = O((\bmax+\ell^\prime)(n+m)\log \ n \cdot \epsilon^{-2}).
\end{align}

Notice that generating RR set from scratch for the final node selection, following the fix of \cite{chen2018issue}, only adds a multiplicative factor of $2$. Hence the overall asymptotic running time to generate $\mathcal{R}$ remains unaffected. Thus intuitively, there are two changes in $\PRIMM$'s running time. 
The budget $k$ of a single item of IMM is replaced with $\bmax$, the maximum budget of any item. Secondly, by applying union bound on every individual item's failure probability, a factor of $\log_n|\bvec|$ is added to the sample complexity.
Using Lemma \ref{lem:primmcorrect} and Eq.~\eqref{eq:rtime} we now prove the
	correctness and the running time result of $\PRIMM$.
	
\begin{restatable}{theorem}{thmPRIMM} \label{thm:PRIMM}
	$\PRIMM$ is prefix preserving and returns a $(1 - 1/e - \epsilon)$-approximate solution to IM with at least $1 - 1/n^\ell$ probability in $O((\bmax+\ell+\log_n |\bvec|)(n+m)\log \ n \cdot \epsilon^{-2})$ expected time.
\end{restatable}

\begin{proof}
	From Lemma \ref{lem:primmcorrect}, we have that $\PRIMM$ returns a prefix preserving $(1 - 1/e - \epsilon)$-approximate solution with at least $1 - 1/n^\ell$ probability. In that process $\PRIMM$ invokes $\nodeselect$, $\log_2 \ n -1$ times in the while loop and once to find the final seed set $S_{\bmax}$. 
	Note that, we intentionally avoid redundant calls to $\nodeselect$ when we switch budgets,
	which saves $|\bvec| $ additional calls to $\nodeselect$.
	
	Let $\mathcal{R}_i$ be the susbset of $\mathcal{R}$ used in the $i$-th iteration of the loop. Since $\nodeselect$ involves one pass over all RR set, on a given input $\mathcal{R}_i$, it takes $O(\sum_{R \in \mathcal{R}_i} |R|)$ time. Recall $|\mathcal{R}_i|$ doubles with every increment of $i$. Hence it is a geometric sequence with a common ratio of $2$. Now from Theorem $3$ of \cite{tang15} and the fact that there is no additional calls to $\nodeselect$ 
	during budget switch, we have total cost of invoking all $\nodeselect$ is $O( \mathbb{E}[\sum_{R \in \mathcal{R}}|R|])$.
	
	Since $|R| \leq w(R)$, for any $R \in \mathcal{R}$, then using Eq. \eqref{eq:rtime} we have,
	\begin{align*}
	O(\mathbb{E}[\sum_{R \in \mathcal{R}}|R|]) &= O(\mathbb{E}[\sum_{R \in \mathcal{R}}w(R)]) \\
	& = O((\bmax+\ell^\prime)(n+m)\log \ n \cdot \epsilon^{-2}) \\
	& = O((\bmax+\ell + \log_n |\bvec|)(n+m)\log \ n \cdot \epsilon^{-2}).
	\end{align*}
	Hence the theorem follows.
\end{proof}

Finally, the combination of Theorems~\ref{thm:nonuniformbudget} and~\ref{thm:PRIMM}
	gives our main Theorem~\ref{thm:main}.

\subsection{Experiments}\label{sec:exp}
\begin{table}[t!]
	
	\centering
	\begin{tabular}{rccccc}
	 & \flix & \dbBook & \dbMovie & \twit & \orkut \\ \hline
	{ \# nodes}			& $7.6$K & $23.3$K	& $34.9$K & $41.7$M & $3.07$M  \\ 
	{ \# edges} 		& $71.7$K  & $141$K	& $274$K & $1.47$G & $234$M  \\ 
	 { avg. degree}     & $9.43$ & $6.5$ & $7.9$	& $70.5$  & $77.5$\\ 
	 { type}            &  undirected  & directed & directed  &  directed & undirected  \\ \hline
	\end{tabular}
	\caption{Network Statistics}
	 \label{tab:datasets}
\end{table}
\subsubsection{Experiment Setup}

\begin{figure*}[ht]
\hspace*{-2mm}\includegraphics[width=1.05\textwidth]{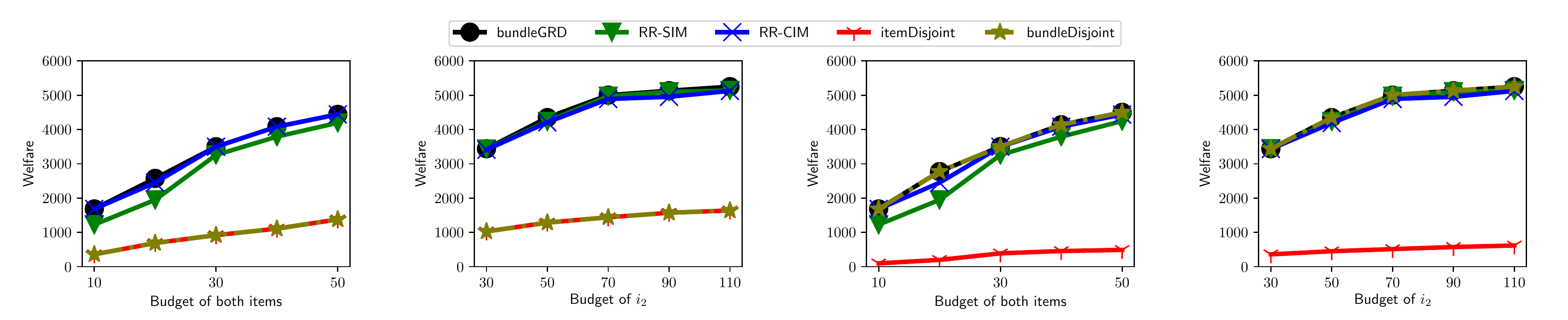} 
\begin{tabular}{cccc}
 \hspace{2mm} (a) Configuration $1$  \hspace{14mm} &  (b) Configuration $2$ \hspace{13mm} &  (c) Configuration $3$ \hspace{12mm} &  (d)  Configuration $4$
\end{tabular}
\caption{Expected social welfare in four configurations (on the \dbMovie network)} \label{fig:welfare}
\end{figure*}

\begin{figure*}[h!t!]
\hspace*{-2mm}\includegraphics[width=1.05\textwidth]{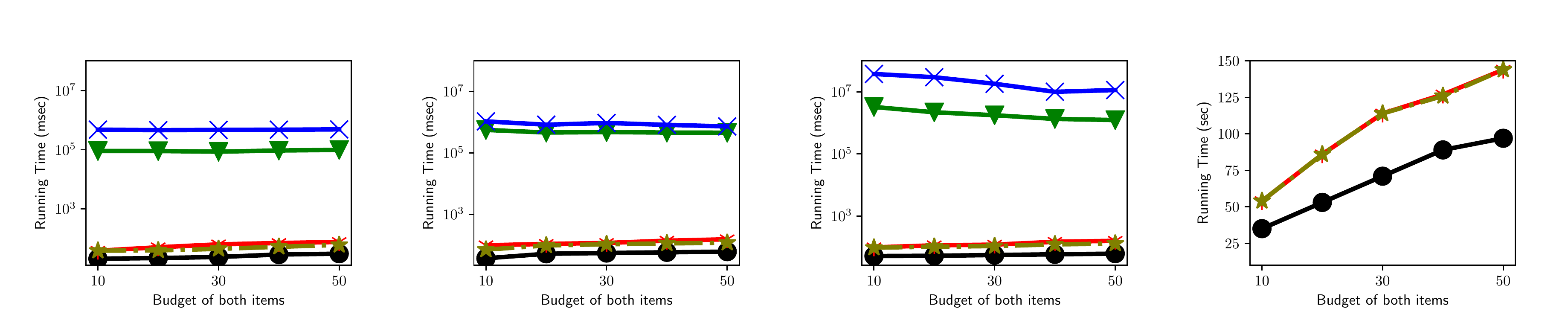}
\begin{tabular}{cccc}
\hspace{6mm} (a)  \flix  \hspace{18mm} &  (b) \dbBook \hspace{12mm} &  (c) \dbMovie \hspace{17mm} &  (d)  \twit
\end{tabular}
\caption{Running times of $\algo$, $\RRSIM$, $\RRCIM$, $\id$ and $\bd$ (on Configuration $1$)} \label{fig:time}
\end{figure*}


\begin{figure*}[ht]
\hspace*{-2mm}\includegraphics[width=1.05\textwidth]{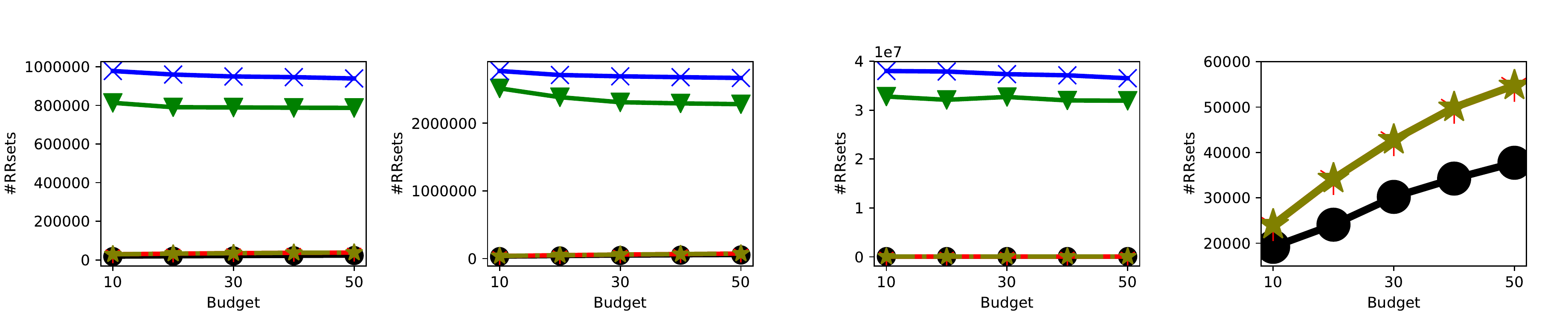}
\begin{tabular}{cccc}
\hspace{6mm} (a)  \flix  \hspace{18mm} &  (b) \dbBook \hspace{12mm} &  (c) \dbMovie \hspace{17mm} &  (d)  \twit
\end{tabular}
\caption{Number of RR sets generated by $\algo$, $\RRSIM$, $\RRCIM$, $\id$ and $\bd$ (on Configuration $1$)} \label{fig:rrsets}
\end{figure*}

\begin{figure*}[h!t!]
\hspace*{-2mm}\includegraphics[width=1.05\textwidth]{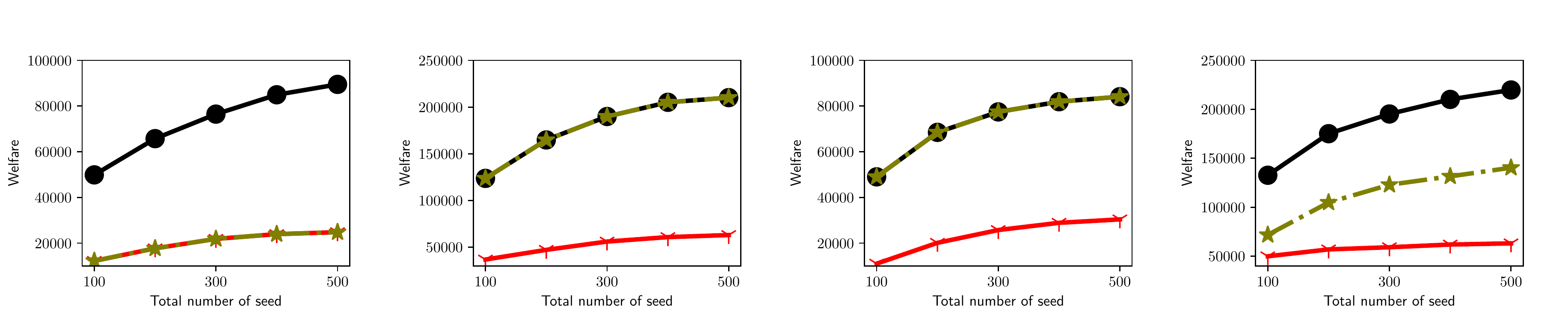} 
\begin{tabular}{cccc}
 \hspace{2mm} (a) Configuration $5$ & \hspace{14mm} (b) Configuration $6$   & \hspace{13mm} (c) Configuration $7$  & \hspace{12mm} (d) Configuration $8$
\end{tabular}
\caption{Expected social welfare in four configurations (on the \twit network)}\label{fig:synsw} 
\end{figure*}

\begin{figure*}[h!t!]
\hspace*{-2mm}\includegraphics[width=1.05\textwidth]{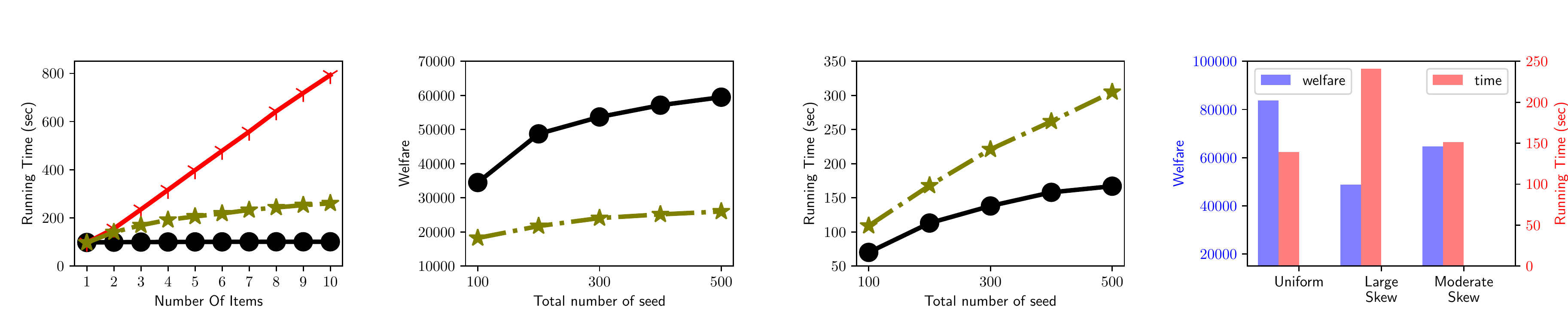} 
\begin{tabular}{cccc}
 \hspace*{-2mm} (a) Effect of number of items & \hspace*{10mm} (b) Welfare   & \hspace*{16mm} (c) Running time  & \hspace*{18mm} (d) Budget skew
\end{tabular}
\caption{(a) Impact of number of items on the running time and (b-d) Experiments using real $\parameterset$ (on the \twit network)}\label{fig:real} 
\end{figure*}

\begin{figure*}[h!t!]
\hspace*{-21mm}\includegraphics[width=1.25\textwidth]{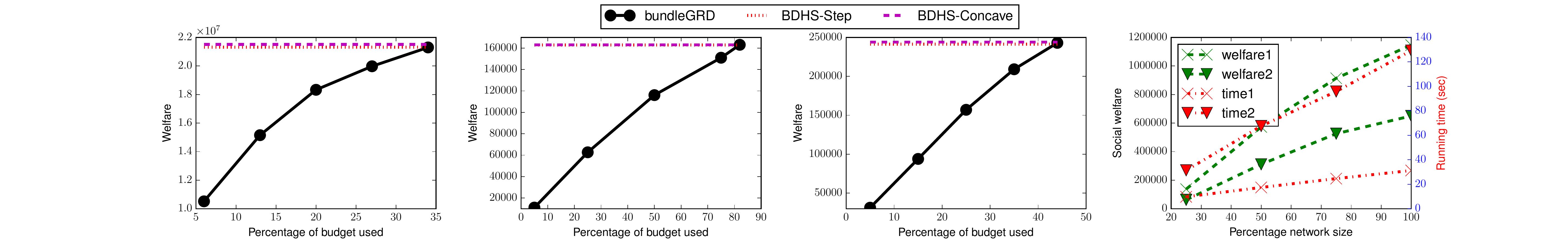}
\begin{tabular}{cccc}
\hspace{2mm} (a)  \orkut  \hspace{20mm} &  (b) \dbBook \hspace{12mm} &  (c) \dbMovie \hspace{19mm} &  (d)  \orkut
\end{tabular}
\caption{(a-c) Comparison against BDHS algorithms and (d) Scalability of $\algo$ } \label{fig:tcs}
\end{figure*}

\eat{
\begin{figure*}[ht]
\begin{minipage}[b]{0.45\linewidth}
\includegraphics[width=.48\textwidth]{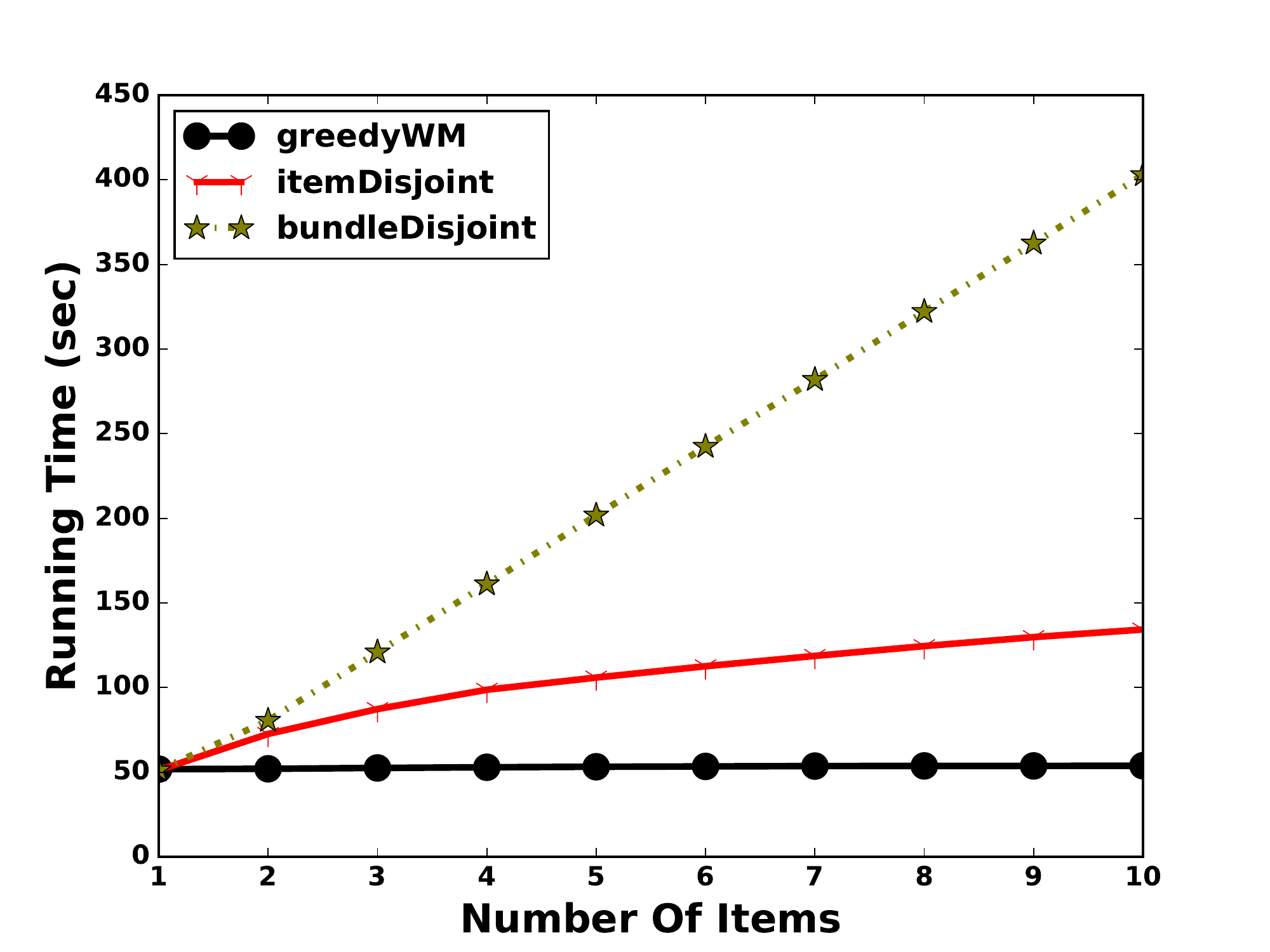}
\caption{Effect of number of items}
\label{fig:multiitemrun}
\end{minipage}
\hspace{0.5cm}
\begin{minipage}[b]{0.45\linewidth}
\begin{small}
\centering
\begin{tabular}{ccc}
\hspace{-12mm}\includegraphics[height=3cm,width=1.05\textwidth]{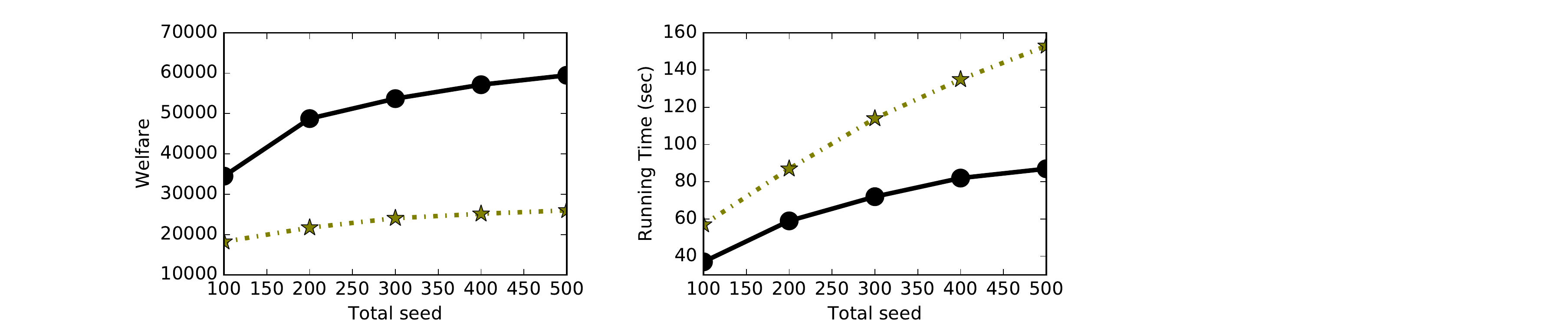} &
\hspace{8mm}\includegraphics[height=3cm,width=.31\textwidth]{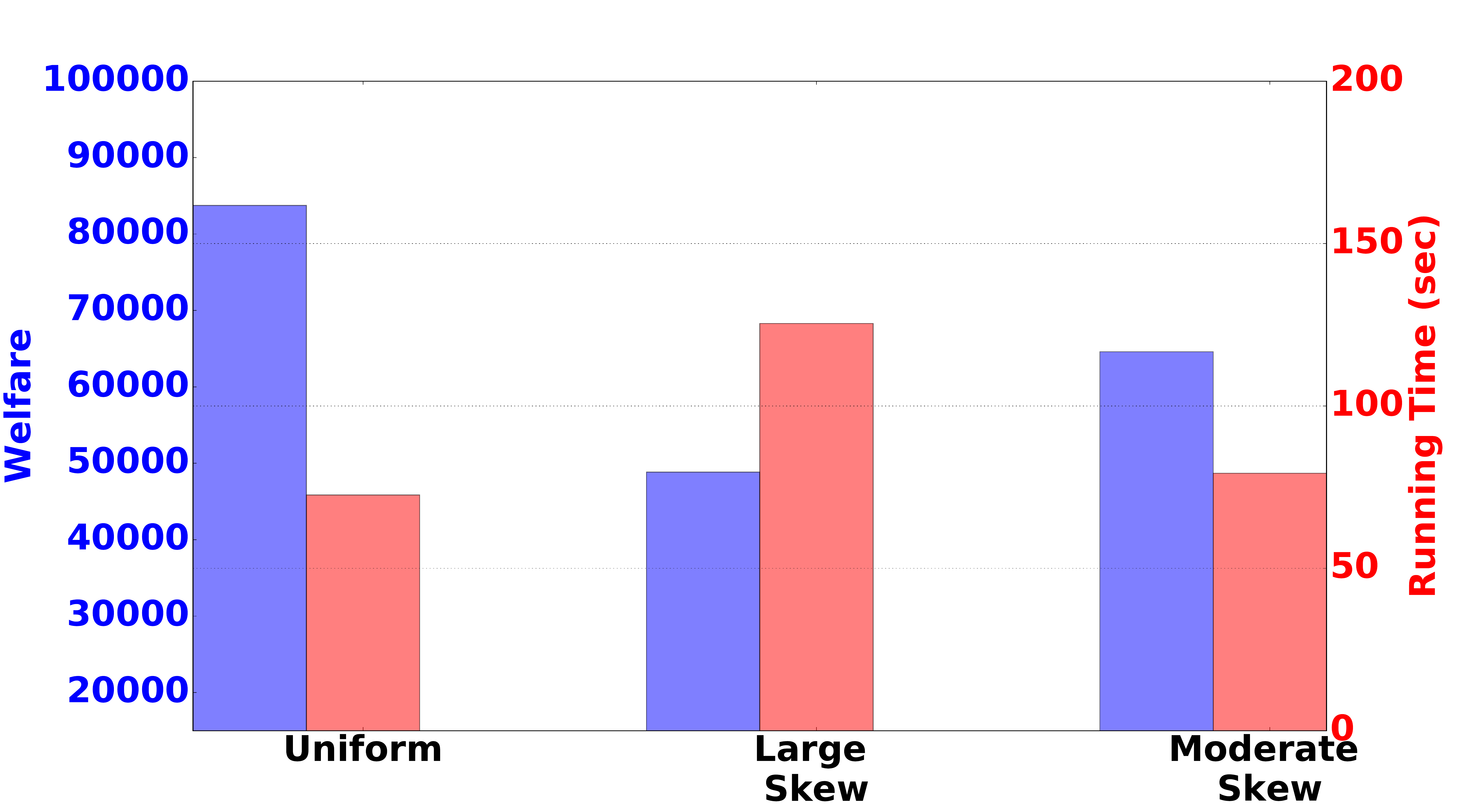}\\
\hspace{-53mm} (a) Social welfare  \hspace{30mm} (b) Running time   & \hspace{-68mm} (c) Effect of budget distribution on $\algo$
\end{tabular}
\caption{Results on real dataset}\label{fig:real} 
\end{small}
\end{minipage}
\end{figure*}
}
 
We perform extensive experiments on five real social networks.
We first experiment with synthetic utility (value and price) functions. 
For real utility functions, we learn the value and noise distributions of items from the bidding data in \ebay, 
and obtain item prices from \cgl and \fb groups \laksRev{to make them compatible with used items auctioned in eBay}. All experiments are performed on a Linux machine with Intel Xeon $2.6$ GHz CPU and $128$ GB RAM.

\spara{4.3.1.1 Networks} 
Table \ref{tab:datasets} summarizes the networks and their characteristics. 
\flix is mined in \cite{lu2015} from a social movie site and a strongly connected component is extracted. \db is a Chinese social network, where users rate  books, movies, music, etc.
In \cite{lu2015} all movie and book ratings of the users in the graph are crawled separately to derive two datasets from book and movie ratings: \dbBook and \dbMovie. \twit is one of the largest public network datasets. 
Finally \orkut is a large social network that we use to test scalability. Both \twit and \orkut can be obtained from \cite{twitter}.

\spara{4.3.1.2 Algorithms compared}
We compare $\algo$ against six baselines -- $\RRSIM$, $\RRCIM$, $\id$, $\bd$, $\tcsc$ and $\tcss$. $\RRSIM$ and $\RRCIM$ are two state-of-the art algorithms designed for complementary products in the context of IM \cite{lu2015}. However, they work only for two items. Extending the \comic framework and the $\RRSIM$ and $\RRCIM$ algorithms for more than two items is highly non-trivial as that requires dealing with automata with exponentially many states. 
Hence in comparing the performance of $\algo$ against $\RRSIM$ and $\RRCIM$, we  limit the number of items to two. Later we experiment with more than two items. Below, by deterministic utility of an itemset $I$, we mean $\val(I) - \price(I)$, i.e., its utility with the noise term ignored. 

\begin{enumerate}

\item {\bf \comic baselines.} For two items $\ia$ and $\ib$, given seed set of item $\ib$ (resp. $\ia$), $\RRSIM$ (resp. $\RRCIM$) finds seed set of item $\ia$ (resp. $\ib$) such that expected number of adoptions of $\ia$ is maximized. Initial seeds of $\ib$ (resp. $\ia$) are chosen using IMM \cite{tang15}. 

\item {\bf Item-disjoint.} Our next baseline $\id$ allocates only one item to every seed node. Given the set of items $\allitems$, $\id$ finds $\sum_{i \in \allitems} b_i$ nodes, say $L$,  using IMM \cite{tang15}, where $b_i$ is the budget of item $i$. Then it visits items in $L$ in non-increasing order of budgets, assigns item $i$ to first $b_i$ nodes and removes those $b_i$ nodes from $L$. By explicitly assigning every item to different seeds, $\id$ does not leverage the effect of supermodularity. However it benefits from the \textit{\chgins{network propagation}}: since the utilities are supermodular, if more neighbors of a node adopt some item, it is more likely that the node will also adopt an item. Thus, when individual items have positive utility and hence can be adopted and propagate on their own, by choosing more seeds,  $\id$ makes use of the \chgins{network propagation} to encourage more adoptions.

\item {\bf Bundle-disjoint.} Baseline $\bd$, aims to leverage both supermodularity and \chgins{network propagation}. It first orders the items $\allitems$ in non-increasing budget order and determines successively minimum sized subsets with non-negative deterministic utility, maintaining these subsets (``bundles'') in a list. Items in each bundle $B$ are allocated to a new set of $b_B := {\it min\/}\{b_i\mid i\in B\}$ seed nodes. The budget of each item in $B$ is decremented by $b_B$, and items with budget $0$ are removed. When no more bundles can be found, we revisit each item $i$ with a positive unused budget and repeatedly allocate it to the seeds of the first existing bundle $B$ which does not contain $i$. If $b_B > b_i$ (where $b_i$ is the current budget of $i$ after all deductions), then the first $b_i$ seeds from the seed set of $B$ are assigned to $i$. If an item $i$ still has a surplus budget, we select $b_i$ fresh seeds using IMM and assign them to $i$.

\item {\bf Welfare maximization baselines.} Our last two baselines, $\tcsc$ and $\tcss$ are two state-of-the-art welfare maximization algorithms under network externalities \cite{BhattacharyaDHS17}. 
As discussed in \textsection~\ref{sec:related}, their study has significant differences from our study, 
	but we still make an empirical comparison
	with their algorithms with the goal
	to \laksRev{explore what fraction of the  budget is needed by our model with network propagation to achieve the same
	social welfare as their model which has network externality but no network propagation.} 
We defer the details of the comparison method to  \textsection~\ref{sec:real_data}.

\end{enumerate}
%

\begin{table}[]
\hspace{-5mm}
\begin{tabular}{|l|c|c|c|c|l|}
\hline
\multicolumn{1}{|c|}{No} & Price                                                                                  & Value                                                                        & \multicolumn{1}{l|}{Noise}                                                                            & \multicolumn{1}{l|}{GAP}                                                                       & Budget     \\ \hline
1                        & \multirow{4}{*}{\begin{tabular}[c]{@{}c@{}}$\ia=3$\\ $\ib=4$\\ $\{\ia,\ib\}=7$\end{tabular}} & \multirow{2}{*}{\begin{tabular}[c]{@{}c@{}}$\ia=3,\ib=4$\\ $\{\ia,\ib\}=8$\end{tabular}} & \multirow{4}{*}{\begin{tabular}[c]{@{}c@{}}$\ia:N(0,1)$\\ $\ib:N(0,1)$\\ $\{\ia,\ib\}:N(0,2)$\end{tabular}} & \multirow{2}{*}{\begin{tabular}[c]{@{}c@{}}$\qao=0.5$, $\qbo=0.5$\\ $\qab=0.84,\qba=0.84$\end{tabular}} & Uniform    \\ \cline{1-1} \cline{6-6} 
2                        &                                                                                        &                                                                              &                                                                                                       &                                                                                                & Nonuniform \\ \cline{1-1} \cline{3-3} \cline{5-6} 
3                        &                                                                                        & \multirow{2}{*}{\begin{tabular}[c]{@{}c@{}}$\ia=3,\ib=3$\\ $\{\ia,\ib\}=8$\end{tabular}} &                                                                                                       & \multirow{2}{*}{\begin{tabular}[c]{@{}c@{}}$\qao=0.5$, $\qbo=0.16$\\ $\qab=0.98,\qba=0.84$\end{tabular}} & Uniform    \\ \cline{1-1} \cline{6-6} 
4                        &                                                                                        &                                                                              &                                                                                                       &                                                                                                & Nonuniform \\ \hline
\end{tabular}
\caption{Two item configurations}
\label{tab:configs}
\end{table}

\spara{4.3.1.3 Default Parameters}
Following previous works \cite{Huang2017,Nguyen2016} we set probability of edge $e = (u,v)$ to $1/\ind(v)$. 
Unless otherwise specified, we use $\epsilon = 0.5$ and $\ell = 1$ as our default for all five methods as recommended in \cite{tang15,lu2015}. 
The \comic algorithms $\RRSIM$ and $\RRCIM$ use adoption probabilities, called GAP parameters \cite{lu2015}, to model the interaction between items. The GAP parameters can be simulated within the \model framework using utilities shown in Eq. \eqref{eq:gaputil}. The derivation follows simple algebra. Here, $\qao$ (resp., $\qab$) denotes the probability that a user adopts
item $\ia$ given that it has adopted nothing (resp., item $\ib$).

Let $\ia$ and $\ib$ be the two items. Suppose the desire set of a node $u$ only has item $\ia$. The condition that $u$ adopts $\ia$ is
	$\val(\ia)-\price(\ia)+\noise(\ia) \ge 0$.
Thus the GAP parameter $\qao$ is given by:
\begin{equation*}
\qao = \Pr[ \val(\ia)-\price(\ia)+\noise(\ia) \ge 0] = \Pr[ \noise(\ia) \ge \price(\ia) - \val(\ia)].
\end{equation*}

Now suppose $\ia$ has been adopted by $u$, and $\ib$ enters the desire set. 
The GAP parameter $\qba$ is the probability of adopting $\ib$ given that $\ia$ has been adopted. So we have
\begin{align*}
\qba & = \Pr[\val(\{\ia,\ib \})-\price(\ia) - \price(\ib) + \noise(\ia) + \noise(\ib) \ge \\
	& \ \  \hspace*{2ex}  \val(\ia) - \price(\ia) + \noise(\ia)
\mid \noise(\ia) \ge \price(\ia) - \val(\ia) ] \\
 & = \Pr[  \noise(\ib) \ge \price(\ib) - (\val(\{\ia,\ib \} ) - \val(\ia)) \mid \noise(\ia) \ge \\ 
 & \ \hspace*{2ex} \price(\ia) - \val(\ia) ].
\end{align*}
Since noise $\noise(\ib)$ is independent of noise $\noise(\ia)$, we can remove the above condition in the
	conditional probability, and obtain
\begin{equation*}
\qba =  \Pr\{  \noise(\ib) \ge \price(\ib) - (\val(\{\ia,\ib \} ) - \val(\ia))\}.
\end{equation*}
The other two GAP parameters, $\qbo$ and $\qab$ can be obtained similarly.
To summarize, we have

\begin{align} \label{eq:gaputil}
\qao & = \Pr[ \noise(\ia) \ge \price(\ia) - \val(\ia)], \nonumber \\ 
\qab & = \Pr[  \noise(\ia) \ge \price(\ia) - (\val(\{\ia,\ib \} ) - \val(\ib))], \\ \nonumber
\qbo & = \Pr[ \noise(\ib) \ge \price(\ib) - \val(\ib)], \\ \nonumber
\qba & =  \Pr[  \noise(\ib) \ge \price(\ib) - (\val(\{\ia,\ib \} ) - \val(\ia))].
\end{align}

\subsubsection{Experiments on two items} \label{sec:two_items}

\eat{ 
To enable comparison with the \comic algorithms $\RRSIM$ and $\RRCIM$, in our first set of experiments, we limit number of items to two.
} 
We explore four different configurations corresponding to the choice of the values, prices, noise distribution parameters, and item budgets (see Table \ref{tab:configs}). While \model does not assume any specific distribution for noise, in our experiments we use a Gaussian distribution for illustration.

In Configurations  $1$ and $2$, individual items have non-negative deterministic utility. In this setting $\id$ and $\bd$ are equivalent. In Configurations $3$ and $4$ one item has a negative deterministic utility while the other item has a non-negative one. In this setting, however, $\algo$ and $\bd$ are equivalent. One may also consider configurations where every individual item has negative deterministic utility. In such a  setting, $\id$ produces $0$ welfare, which makes the comparison degenerates.

For every parameter setting, we consider two budget settings, namely \textit{uniform} (e.g., Configuration $1$) and \textit{non-uniform} (resp. Configuration $2$). In case of uniform budget, both items have the same budget $k$, where $k$ is varied from $10$ to $50$ in steps of $10$. For non-uniform budget, $\ia$'s budget is fixed at $70$, and $\ib$'s budget is varied from $30$ to $110$ in steps of $20$.

\spara{4.3.2.1 Social Welfare}
We compare the expected social welfare achieved by all algorithms on all four configurations (Fig. \ref{fig:welfare}). We show the results only for \dbMovie, since the trend of the results is similar on other networks.
In terms of social welfare, $\algo$ achieves an expected social welfare upto $5$ times higher than $\id$ (Fig. \ref{fig:welfare}(d)). 
\eat{
\note[Laks]{Higher than what? Also, the figures show that $\algo$ is tied with one more algorithms w.r.t. SW. So, $\algo$ is NOT ouperforming all the baselines consistently. If so, what point are we making with these plots: why should $\algo$ be preferred? BTW, the label geedyWM in plots should be consistently changed to $\algo$.} 
} 
A similar remark applies when $\bd$ and $\algo$ are not equivalent (e.g., Fig. \ref{fig:welfare}(b)). Further, notice that $\RRSIM$ and $\RRCIM$ produce welfare similar to $\algo$. It follows from 
Table $4$ of \cite{lu2015} (full arxiv version) that under this configuration, $\RRSIM$ and $\RRCIM$ end up copying the seeds of the other item. Hence their allocations are similar to $\algo$.  However, as shown next, $\algo$ is much more efficient than the other two algorithms, and easily supports more than two items, which makes $\algo$ more suitable in practice for multiple items over large networks.

\spara{4.3.2.2 Running time}
We study the running time of all algorithms using Configuration $1$ as a representative case. 
The results are shown in Fig. \ref{fig:time}. As can be seen, $\algo$ and $\bd$ are equivalent and hence have the same running time. However, $\algo$ significantly outperforms all other baselines on every dataset. $\RRSIM$ and $\RRCIM$ are particularly slow. In fact, on the large $\twit$ network, they could not finish even after our timeout after 
	$6$ hours (hence they are omitted from Fig. \ref{fig:time}(d)). In comparison with the baselines, $\algo$ is upto $5$ orders of magnitude (resp. $1.5$ times) faster than $\RRCIM$ (resp. $\id$). Running times on other configurations show a similar trend, and are omitted. 

\spara{4.3.2.3 Memory}
Lastly we study the memory required by all algorithms using Configuration $1$. 
The results are shown in Fig. \ref{fig:rrsets}. Since the amount of memory required is directly related to the number of RR sets each algorithm produces, we show the RR set numbers in the plots. $\RRSIM$ and $\RRCIM$ are based on TIM, whereas the other three algorithms leverage IMM, which generates much less number of RR sets than TIM. Further for the comic algorithms, the two separate pass involving forward and backward simulations also results in more RR sets generation.

\subsubsection{More than two items}\label{sec:mtsyn}

We use the largest dataset Twitter for tests in this subsection.

\spara{4.3.3.1 The configurations}  
Having established the superiority of $\algo$ for two items, we now consider more than two items. 
Recall that $\RRSIM$ and $\RRCIM$ cannot work with more than two items, so we confine our comparison to $\id$ and $\bd$. 
\eat{Again several interesting configurations can be obtained based on the budgets and the utilities of the items. We devise four such specific configurations. Unless otherwise stated, in every configuration the number of items is assumed to be $10$.
} 
We gauge the performance of the algorithms on social welfare and running time. We also study the effect of budget distribution on social welfare. We design four configurations corresponding to the choice of budget and utility (see Table~\ref{tab:configs-2+items}). For all configurations, we sample noise terms from $N(0,1)$. Price and value are set in such a way as to achieve certain shapes for the set of itemsets in the lattice that have a positive utility (see below). 

\begin{itemize}

\item {\bf Configurations 5-7.} Configuration 5 is the simplest: every item has the same budget; price and value are set such that every item has the same utility of $1$ and utility is additive. Thus, by design, this configuration gives minimal advantage to any algorithm that tries to leverage supermodularity.  The next two configurations (6 and 7) model the situation where a single ``core'' item is necessary in order to make an itemset's utility positive. E.g., a smartphone may be a core item, without which its accessories do not have a positive utility. We set the core item's utility to $5$. The addition of any other item increases the utility by $2$. Thus, all supersets of the core item have a positive utility, while all other subsets have a negative utility. Hence, the set of subsets with positive utility forms a ``cone'' in the itemset lattice. In Configuration 6 (resp. 7), the core item is the item with maximum (resp. minimum) budget. Finally, we design a more general configuration where the set of itemsets with positive utility forms an arbitrary shape (see Configuration 8 below). 

\item {\bf Configuration 8.} We consider the itemset lattice, with level $t$ having subsets of size $t$. We randomly set the prices and values of items in level 1 such that a random subset of items have a non-negative utility. Let $A_t$ be any itemset at level $t > 1$ and $i \in A_t$ any item. We choose a value uniformly at random, $\epsilon \sim U[1,5]$, and define 

\begin{equation}\label{eq:mar_val}
\val(i|A_t \setminus \{i\}) = max_{B \in \power(A_t \setminus \{i\}, t-2)} \{\val(i|B) + \epsilon\}
\end{equation}

where $\power(A, q)$ denotes the set of subsets of $A$ of size $q$. That is, the marginal gain of an item $i$ w.r.t. $A_t\setminus \{i\}$ is set to be the maximum marginal gain of $i$ w.r.t. subsets of $A_t$ of size $t-2$, plus a randomly chosen boost ($\epsilon$). 
E.g., let $A_4=\{i,j,k,l\}$, $t = 4$ then, $\val(i|\{j,k,l\}) = max\{\val(i|\{j,k\}), \val(i|\{k,l\}), \val(i|\{j,l\})\} + \epsilon$.

Recall that the value computation proceeds level-wise starting from level $t=0$. Thus, for any itemset $A_t$ in Eq.\eqref{eq:mar_val}, $\val(i|B)$ for subsets $B$ is already defined. 

\begin{table}

\begin{tabular}{|c|c|c|} \hline
  No & Value & Budget \\ \hline 
  5 & Additive & Uniform \\ \hline
  6 & Cone-max & Non-uniform \\ \hline
  7 & Cone-min & Non-uniform \\ \hline
  8 & Level-wise & Uniform \\ \hline
  \end{tabular}
\caption{Multiple item configurations}
\label{tab:configs-2+items}

\end{table}

Finally, we set 
$\val(A_t) = max_{i \in A_t} \{\val(A_t \setminus \{i\}) + \val(i|A_t \setminus \{i\})\}$. Now we show that this way of assigning values ensures that the value function is well-defined and supermodular. 
\end{itemize}

\begin{lemma}
The value function of Configuration 8 is supermodular.
\end{lemma}
\begin{proof}
First we show that for an itemset $A_t$ at level $t$, and an item $i \notin A_t$, $\val(i \mid A_t) \geq \val(i \mid B)$, where $B \subset A_t$ is any  subset of $A_t$. We prove this claim by induction on level.

\noindent 
	\underline{Base Case}: Let $t=1$ and $A_1$ be any singleton itemset.  Then $\val(i \mid A_1) = \val(i \mid \emptyset) + \epsilon \geq \val(i \mid \emptyset)$.
	
	\noindent 
	\underline{Induction}: Suppose the claim is true for all levels $t \leq l$. We show it holds for $t = l+1$. From our method of assigning values we have, $\val(i \mid A_{l+1}) = \max_{B_l \in \power(\allitems, l)} \{\val(i \mid B_l)\} + \epsilon$, where $\power(\allitems, l)$ is the set of all itemsets at level $l$. Thus $\val(i \mid A_{l+1}) \geq \val(i \mid B_l)$. By induction hypothesis,  $\val(i\mid B_l) \geq \val(i \mid B)$, for any 
subset $B \subset B_l$, and thus 
$\val(i \mid A_{l+1}) \geq \val(i \mid B)$.
	
It then follows that  for any itemsets $B \subset A \subset \allitems$ and item $i \in \allitems \setminus A$, $\val(i \mid A) \geq \val(i \mid B)$. Hence value is supermodular. 
\end{proof}

\begin{lemma}
The value function of Configuration 8 is well defined.
\end{lemma}
\begin{proof}
We show that for an itemset $A_t$ at level $t$, $\val(i \mid A_t \setminus \{i\}) + \val(i) = \val(j \mid A_t \setminus \{j\}) + \val(j) $, for any $i , j \in A_t$.

Let $\val(A_t) = max_{k \in A_t} \{\val(A_t \setminus \{k\}) + \val(k|A_t \setminus \{k\})\} = m$. Then according to our configuration  $\val(i \mid A_t \setminus \{i\}) = m - \val(i)$. Similarly $\val(j \mid A_t \setminus \{j\}) = m - \val(j)$. Hence $\val(i \mid A_t \setminus \{i\}) + \val(i) = m - \val(i) + \val(i) = m - \val(j) + \val(j) = \val(j \mid A_t \setminus \{j\}) + \val(j)$.
\end{proof}

\spara{4.3.3.2 Social welfare}
First, we study the social welfare achieved by the algorithms, in each of the above configurations, with the total budget varying from $500$ to $1000$ in steps of $100$. For Configurations 7 and 10, we set the budget uniformly for every item. 
For other configurations, the max budget is set to $20\%$ of the total budget,  min budget to $2\%$, and the remaining budget is split uniformly. The results of the experiment on \twit network are shown in Fig. \ref{fig:synsw}. Under Configurations 8 and 9, $\algo$ and $\bd$ produces the same allocation, hence the welfare is the same. However in general $\algo$ outperforms every baseline in all the four configurations by producing welfare up to $4$ times higher than baselines.

\spara{4.3.3.3 Running time vs number of items}
Next, we study the effect of the number of items on the running time of the algorithms. For this experiment, we use Configuration 5. We set the budget of every item to $k=50$ and vary the number of items $s$, from $1$ to $10$. Fig. \ref{fig:real}(a) shows the running times on the $\twit$ dataset. As the number of items increases the number of seed nodes to be selected for $\id$ and $\bd$ increases. Notice both $\id$ and $\bd$ select the same number of seeds, which is $k \times s$. $\id$ selects it by one invocation of IMM, with budget $ks$, while $\bd$ invokes IMM $s$ times with budget $k$ for every invocation. So their overall running times differ.  By contrast, the running time of $\algo$ only depends on the maximum budget and is independent of the number of items. E.g., when number of items is $10$, $\algo$ is about $8$ times faster than $\bd$ and $2.5$ times faster than $\id$. 

\subsubsection{Experiment with real value, price, and noise parameters} \label{sec:real_data}
In this section, we conduct experiments on parameters (value, price, and noise) learned from real data. We consider the following 5 items: (1) Playstation $4$, $500$ GB console, denoted $\ps$, (2) Controller of the Playstation, denoted $\cs$, and (3-5) Three different games compatible with $\ps$, denoted $\ga$, $\gb$ and $\gc$ respectively. We next describe the method by which we learn their parameters from real data.

\spara{4.3.4.1 Learning the value, price, and noise}
{\color{black}
Predicting a user's bid in an auction is a widely studied problem in auction theory. Jiang et al. \cite{Albertbidding07} showed that learning user's valuations of items improves the prediction accuracy. Given the bidding history of an item, their method learns a value distribution of the item, by taking into account hidden/unobserved bids. 
We use it to learn the values of itemsets from bidding histories. Recall that in our model value is not random, instead noise models the randomness in valuations. Hence \chgins{we take the mean of the learned distribution to be the value and the noise is set to have $0$ mean and the same variance as the learnt distribution. While \model does not assume specific noise distributions, for concreteness, we fit a Gaussian distribution to noise.} 
}
We take $10,000$ independent random samples from the learnt distribution to fit the gaussian.

\begin{table}[h]
\hspace{-3mm}
	
	\begin{tabular}{|c|c|c|c|c|}	
 \hline
 Itemset & Price & Value & Noise & eBay bidding link \\
 \hline
 $\{\ps\}$ & $260$ & $213$ & $N(0,4)$ & https://ebay.to/2ym9Ioj  \\
 \hline
 $\{\ps,\cs\}$ & $280$ & $220$ & $N(0,6)$ & https://ebay.to/2Escb68  \\
 \hline
 $\{\ps,\ga,\gb,\gc\}$ & $275$ & $258$ & $N(0,4)$ & https://ebay.to/2QYpmxh  \\
 \hline
 $\{\ps,\ga,\gb,\cs\}$ & $290$ & $292.5$ & $N(0,5)$ & https://ebay.to/2ClEnF2  \\
 \hline
 $\{\ps,\ga,\gb,\gc,\cs\}$ & $295$ & $302$ & $N(0,7)$ & https://ebay.to/2P60y99  \\
 \hline
	\end{tabular}
	\caption{Learned parameters}
	 \label{tab:realp}
\end{table}

We mine the bidding histories of different itemsets from \ebay. 
To match the used products bidden in \ebay, we use prices for the used products on
	\cgl and \fb groups.
	Since the items bidden in \ebay are typically used products, to match them with the right price information, we use \cgl and \fb groups where the exact same old product is sold.

The price obtained is C\$$260$ for $\ps$, C\$$20$ for $\cs$, and C\$$5$ each for $\ga,\gb$ and $\gc$. For some of the itemsets, we show the learned parameters and the links to the corresponding \ebay bidding histories used in the learning, in Table \ref{tab:realp}. The rest of the itemsets are omitted from the table for brevity. 
We describe the parameters of those omitted itemsets here. Firstly, any of $\cs,\ga,\gb,\gc$, without the core item $\ps$, is useless. Hence values of those items are set to $0$. Secondly, we did not find any bidding record for an itemset consisting of $\ps,\cs$ and a single game. This is perhaps because typically owners of $\ps$ own multiple games and while selling they sell all the games together with $\ps$. Hence, we consider the itemset with $\ps,\cs$ and a single game to have negative deterministic utility. However, as the table shows, itemsets with $\ps,\cs$ and two games have non-negative deterministic utility. Finding the bidding history for the exact same games is difficult, so since games $g_1$--$g_3$ are priced similarly and valued similarly by users, we assume that any itemset with $\ps,\cs$ and any two games has the same utility as that shown in the fourth row of Table \ref{tab:realp}. \chgins{From the value column, we can see that the items indeed follow supermodular valuation, confirming that in practice complementarity arises naturally.} Lastly, the only itemsets that have positive deterministic utility are itemsets with $\ps,\cs$ and at least two games. All other itemsets including the singleton items, have negative deterministic utility. Consequently, we know that the allocation produced by $\id$ will have $0$ expected social welfare, so we omit $\id$ from our experiments, discussed next.

\spara{4.3.4.2 Effect of total budget size}
We compare $\algo$ with $\bd$ on the Twitter dataset with different sizes of total budgets.  
Given a total budget, we assign $30\%,30\%,20\%,10\%,10\%$ of that to $\ps,\cs,\ga,\gb,\gc$ respectively. Then we vary the total budget from $100$ to $500$ in steps of $100$. Fig. \ref{fig:real}(b) shows the welfare: as can be seen, $\algo$ outperforms $\bd$ in both high and low budgets. In fact with higher budget, $\algo$ produces welfare more than $2$ times that of $\bd$. Next we report the running time of the two algorithms in Fig. \ref{fig:real}(c). Since $\bd$ makes multiple calls to IMM, its running time is $1.5$ times higher than $\algo$.

\spara{4.3.4.3 Effect of different item budget given the same total budget}
Our next experiment 
studies the following question. Suppose we have a fixed total budget which we must be divided up among various items. How would the social welfare and running time vary for different splits? Since we have seen that in terms of social welfare $\algo$ dominates all baselines, we use it to measure the welfare. 
Given a total budget of $500$, we split it across $5$ items following three different budget distributions, namely (i) Uniform: each item has the same budget $100$, (ii) Large skew: one item, $\ps$ has $82\%$ of the total budget and the remaining $18\%$ is divided evenly among the remaining $4$ items; and (iii) Moderate skew: Budgets of the $5$ items, $[\ps,\cs,\ga,\gb,\gc]$, are given by the budget vector $\bvec = [150,150,100,50,50]$.

Fig. \ref{fig:real}(d) shows the expected social welfare and the running time of $\algo$ under the three budget distributions on the Twitter dataset. The welfare is the highest under uniform and worst under large skew, with moderate skew in between. 
Running time shows consistent trend, with uniform being the fastest and large skew being the slowest. The findings are consistent with the observation that with large skew, the number of seeds to be selected increases and the allocation cannot take full advantage of supermodularity. 

\spara{4.3.4.4 Effect of propagation vs. network externality}
	We next compare our $\algo$ against the other two baselines, $\tcsc$ and $\tcss$ (referred to as BDHS algorithms for simplicity). $\tcsc$ and $\tcss$ correspond to the concave and step externality algorithms respectively (i.e. Alg 1 and 3 of \cite{BhattacharyaDHS17}). 
	Our overall approach is, despite the differences between our model and BDHS model as highlighted in \textsection\ref{sec:related}, we try to convert our model in a reasonable way to their model by means of restriction, and use their algorithms
	to find the total social welfare that they can achieve.
	Then we gradually increase the budget of items in our model to see at which budget the social welfare achieved
	by our solution reaches the social welfare achieved by their solution that has no budget and assigns items to every
	node directly. This would demonstrate the budget savings due to our consideration of network propagation.

	We now describe how we convert our model to their model.
	First, our model uses network propagation with the UIC model while their model uses network externality without
	propagation.
	To align the two models, we try two alternatives.
	The first alternative is to sample 10,000 live-edge graphs, and the propagation on one live-edge graph bears 
	similarity with the $1$-step function, and thus we use $1$-step externality function on each live-edge graph to compute
	the total social welfare and then average over all live-edge graphs.
	We refer to this alternative $\tcss$.
	The second alternative works when we restrict our UIC model such that every edge has the same propbability $p$.
	In this case, the activation probability of a node $v$ is $1-(1-p)^k$, where $k$ is the number of active neighbors
	of $v$ which is at most the size $s$ of its	2-neighborhood support set.
	This resembles the concave function case in the BDHS model, and thus we use the concave function 
	$1-(1-p)^s$ in their 2-hop model. We refer to this alternative $\tcsc$.

	Second, to align their unit demand model with our model, we treat
	each item subset as a virtual item in their model, so that they can assign item subsets as one virtual item to the
	nodes. 
	Finally, their model has no budget, so they are free to assign all item subsets to all nodes.
	We use this as a benchmark of the total social welfare they can achieve, and see at what fraction of the budget
	we can achieve the same social welfare due to the network propagation effect.

	We used the $\orkut$ as one of the large networks in this study, which also enables the study of the performance of $\algo$ on a large network other than $\twit$ (which is already used in Figure~\ref{fig:time}(d),~\ref{fig:synsw}, and~\ref{fig:real}).Fig. \ref{fig:tcs}(a-c) shows the results on $\orkut$, $\dbBook$ and $\dbMovie$ networks respectively.  \laksRev{The $x$ axis shows the fraction of the budget needed by $\algo$, where 100\% corresponds to a budget of $n$, i.e., \#nodes in the network, which corresponds to the setting of \cite{BhattacharyaDHS17}.} As can be seen, for dense networks like $\orkut$, $\algo$ needs less than $35\%$ as the budget. We found a similar result on $\flix$, not included here for the lack of space. For a sparse graph like $\dbBook$ it needs ~$82\%$, which is still less than the budget of BDHS. Further, since propagation has a submodular growth, much of the budget is used to increase the latter half of the welfare. E.g., even on $\dbBook$, $75\%$ of BDHS' welfare is obtained by only using $50\%$ budget. 
	This test clearly demonstrates that our $\algo$ could leverage the power of propagation, compared to the BDHS
	approach that only considers externality.

\spara{4.3.4.5 Scalability test}
Our next experiment shows the impact of network size on $\algo$ using $\orkut$ with two types of edge probabilities: (1) $1/\ind(v)$ and (2) fixed $0.01$.  
	We use a uniform budget of $50$ for all items. We then use breadth-first-search to progressively increase the network size such that it includes a certain percentage of the total nodes. The results are shown in Fig. 8(d). With increasing network size, the running time in both cases roughly has a linear increase, whereas the welfare depicts a sublinear growth. It is worth noticing that even for the entire million-sized network and fixed probability, $\algo$ requires mere $129$ (time 2) seconds to complete, which again attests to its  scalability.

\eat{ 
\prib{I suggest that we remove the following experiment from the main paper now.}
\weic{I agree that the following test is minor, and if we need space, we can remove it, or simply mention it
	in one sentence.} 
}

\begin{table}
	
	\begin{tabular}{|c|c|c|c|} \hline
  Budget distribution & $\algo$ & MAX\_IMM & IMM\_MAX \\ \hline 
  Uniform & $37719$ & $37719$ & $37719$ \\ \hline
  Large skew & $144328$ & $144328$ & $144328$ \\ \hline
  Moderate skew & $50839$ & $50839$ & $50839$ \\ \hline
  \end{tabular}
   \hspace{-7mm} \caption{The number of RR sets generated}
\label{tab:memory-2+items}
\end{table}

\spara{4.3.4.6 Memory usage}
Lastly we assess the memory usage of $\algo$.
Since the main memory usage is on the RR set storage, we evaluate the
	number of RR sets $\algo$ generates in comparison to IMM for the three aforementioned budget distributions. 
Since IMM works only with a single item (i.e., one budget), we consider two variants. In the first variant IMM is invoked with maximum budget, called IMM\_MAX. The second variant iterates over all budgets and reports the budget that generates the maximum number of RR sets, called MAX\_IMM. Notice IMM\_MAX and MAX\_IMM are not equivalent because the number of RR sets generated by IMM is not monotone in budget. The results are shown in Table \ref{tab:memory-2+items}. In all three budget configurations the \emph{numbers}  of RR sets generated by the three algorithms are exactly the same, from which we can conclude that $\algo$ has a similar memory requirement as IMM.



\section{Summary \& Discussion }\label{sec:concl}
\chgins{We propose a novel model combining influence diffusion with utility-driven item adoption, which supports any mix of competing and complementary items. Focusing on complementary items, we study the problem of optimizing expected social welfare.
Our objective function is monotone, but neither submodular nor supermodular. Yet, we show that a simple greedy allocation guarantees a $(1-1/e-\epsilon)$-approximation to the optimum.} Based on this, we develop a scalable approximation algorithm $\algo$, which satisfies an interesting prefix preserving property. With extensive experiments, we show that our algorithm outperforms the state of the art baselines. 
Our results and techniques carry over unchanged to any triggering propagation model \cite{kempe03}. We assumed that price is additive and valuations are supermodular. 
If we use submodular prices, that would further favor item bundling. In this case, utility remains supermodular and our results remain intact. Independently of this, we could study competition using submodular value functions. 
%
\eat{ 
Making one or both of valuations and noise user dependent would make the model more expressive. 
} 
Orthogonally, we can study the UIC model under personalized noise terms. 
It is interesting to study the expected welfare maximization problem in these alternative settings.

\clearpage  
{
\bibliographystyle{plain}
\bibliography{sigmod2019-epic}  
}

\end{document}